\documentclass[journal,10pt]{IEEEtran}

\usepackage{amsmath}
\usepackage{amssymb}
\usepackage{amsfonts}
\usepackage{graphicx}
\usepackage{epsfig}
\usepackage{subfigure}
\usepackage{psfrag}
\usepackage{cite}
\usepackage{latexsym}
\usepackage{url}
\usepackage{color}
\PassOptionsToPackage{bookmarks={false}}{hyperref}

\begin{document}
\title{Energy Group Buying with Loading Sharing for Green Cellular Networks}
\author{Jie Xu,~Lingjie Duan,~and Rui Zhang
\thanks{J. Xu is with the Engineering Systems and Design Pillar, Singapore University of Technology and Design and the State Key Laboratory of Rail Traffic Control and Safety, Beijing Jiaotong University (e-mail:~jiexu.ustc@gmail.com).}
\thanks{L. Duan is with the Engineering Systems and Design Pillar, Singapore University of Technology and Design (e-mail:~lingjie\_duan@sutd.edu.sg).}
\thanks{R. Zhang is with the Department of Electrical and Computer Engineering, National University of Singapore (e-mail: elezhang@nus.edu.sg). He is also with the Institute for Infocomm Research, A*STAR, Singapore.}
}

\maketitle
\begin{abstract}
In the emerging hybrid electricity market, mobile network operators (MNOs) of cellular networks can make day-ahead energy purchase commitments at low prices and real-time flexible energy purchase at high prices. To minimize electricity bills, it is essential for MNOs to jointly optimize the day-ahead and real-time energy purchase based on their time-varying wireless traffic load. In this paper, we consider two different MNOs coexisting in the same area, and exploit their collaboration in both energy purchase and wireless load sharing for energy cost saving. Specifically, we propose a new approach named energy group buying with load sharing, in which the two MNOs are aggregated as a single group to make the day-ahead and real-time energy purchase, and their base stations (BSs) share the wireless traffic to maximally turn lightly-loaded BSs into sleep mode. When the two MNOs belong to the same entity and aim to minimize their total energy cost, we use the two-stage stochastic programming to obtain the optimal day-ahead and real-time energy group buying jointly with wireless load sharing. When the two MNOs belong to different entities and are self-interested in minimizing their individual energy costs, we propose a novel repeated Nash bargaining scheme for them to negotiate and share their energy costs under energy group buying and load sharing. Our proposed repeated Nash bargaining scheme is shown to achieve Pareto-optimal and fair energy cost reductions for both MNOs.
\end{abstract}

\newtheorem{definition}{\underline{Definition}}[section]
\newtheorem{fact}{Fact}
\newtheorem{assumption}{Assumption}
\newtheorem{theorem}{\underline{Theorem}}[section]
\newtheorem{lemma}{\underline{Lemma}}[section]
\newtheorem{corollary}{\underline{Corollary}}[section]
\newtheorem{proposition}{\underline{Proposition}}[section]
\newtheorem{example}{\underline{Example}}[section]
\newtheorem{remark}{\underline{Remark}}[section]
\newtheorem{algorithm}{\underline{Algorithm}}[section]
\newcommand{\mv}[1]{\mbox{\boldmath{$ #1 $}}}

\section{Introduction}

To meet the explosive wireless traffic growth driven by the popularity of new mobile devices (e.g., smart phones and tablets) and new mobile applications (such as social networking), the fifth-generation (5G) cellular technology has recently attracted a lot of research interests from both academia and industry (see, e.g., \cite{AndrewsWhat2014,IToward2014,CaiChe2015}). It is expected that 5G should achieve roughly 1000 times data rate increase as compared to its fourth-generation (4G) counterpart. However, the growing data throughput will lead to large energy consumption and high electricity bills for mobile network operators (MNOs). For example, the total energy cost of China Mobile, a Chinese state-owned telecommunication company, is almost 3 billion US dollars in the year of 2011 \cite{HasanGreen2011}, and is still increasing considerably year by year. Therefore, how to reduce the energy cost of cellular networks while ensuring subscribers' growing communication requirements is very important for the technological and economical success of 5G, and it is predicted in \cite{AndrewsWhat2014} that the energy cost per bit will need to fall by at least 100 times for 5G.

Recently, the electricity grids are also experiencing a paradigm shift from conventional grids to smart grids \cite{FangSmart2012,Rahbar2015}. Unlike conventional grids using fixed energy pricing, the newly deployed smart grids enable grid operators to charge time-varying prices to cope with electricity consumers' time-varying load, thus helping stabilize the energy generation and transmission \cite{FangSmart2012}. Specifically, a hybrid electricity market of smart grids has been successfully implemented in more and more countries (such as United States and Norway \cite{OttExperience2003,PhilpottOptimizing2006}), and such a hybrid market combines a day-ahead energy market and a real-time energy market. In the day-ahead energy market, electricity consumers can make commitment on tomorrow's energy purchase at low prices; whereas in the real-time energy market, they are free from commitment and can flexibly buy energy at high prices or sell back the excessive energy commitment at prices lower than the day-ahead ones. In view of the new hybrid electricity market, it is essential for MNOs to jointly optimize the day-ahead and real-time energy purchase based on their time-varying wireless traffic load, for the purpose of minimizing their energy costs.

In this paper, we consider a scenario with coexisting MNOs in the same area,\footnote{One practical example of such a scenario is the coexistence of different state-owned MNOs (i.e., China Mobile, China Unicom, and China Telecom) in China.} in which each MNO participates in the day-ahead and real-time energy markets as an electricity consumer, and operates a number of cellular base stations (BSs) to serve their respective subscribers. Under such a scenario, MNOs face two key challenges. First, it may happen that one MNO over-commits while another under-commits the day-ahead energy purchase, thus leading to additional expense in real-time energy trading. Second, cellular BSs of each MNO, which are deployed to meet the peak wireless traffic loads, are under-utilized during the non-peak traffic hours. The under-utilization results in energy inefficiency due to the non-transmission power consumption at BSs for maintaining the routine operation and ensuring the coverage \cite{ArnoldPowerConsumptionModeling2010,AuerHow2011,Peng}. To overcome the two challenges, we are motivated to exploit the collaboration benefit in both energy purchase and wireless load sharing among different MNOs for their energy cost saving.

For the purpose of exposition, in this paper we investigate the collaborative operation of two coexisting MNOs in a daily time horizon including a number of time slots. It is assumed that the real-time energy prices and the BSs' wireless traffic demands are time-varying and can only be partially predicted day-ahead with certain prediction errors. Under this setup, we jointly optimize the day-ahead energy commitment and real-time energy trading at the two MNOs together with their wireless load sharing, so as to minimize their energy costs. To our best knowledge, this paper is the first attempt to investigate the cost-efficient operation of cooperative MNOs in the hybrid energy market, and our proposed energy optimization framework for cellular networks by taking into account the new smart grid features is important to help MNOs meet the 100 times energy cost decrease for 5G. The main results of this paper are summarized as follows.

\begin{itemize}
\item {\it Individual energy buying in the hybrid electricity market.} As a benchmark, we first consider the non-cooperative case when two traditional MNOs operate independently to minimize their individual energy costs. We use the two-stage stochastic programming to obtain the optimal day-ahead energy commitment and real-time energy trading for individuals. It is shown that when the day-ahead energy price is larger than the average of the real-time energy buying and selling prices in a particular time slot, it is beneficial for the MNO to under-commit its predicted energy demand for this slot; otherwise the MNO will over-commit.
\item {\it Energy group buying with load sharing for fully-cooperative MNOs.} Next, we consider the fully-cooperative case when the two MNOs belong to the same entity (like Sprint and T-Mobile merged in some states of US \cite{HardySprint}) and minimize their total energy cost. We propose a new approach named energy group buying with load sharing, in which the two MNOs are aggregated as a single group to make the day-ahead and real-time energy purchase, and their BSs share the wireless traffic to maximally turn lightly-loaded BSs into sleep mode. By jointly optimizing the day-ahead and real-time energy group buying as well as the load sharing via the two-stage stochastic programming, the total energy cost of the two MNOs is significantly reduced as compared to the non-cooperative benchmark.
\item {\it Repeated Nash bargaining scheme for self-interested MNOs.}  Finally, when the two MNOs belong to different entities and are self-interested in minimizing their individual energy costs, we propose a novel repeated Nash bargaining scheme for them to negotiate and share their energy costs under energy group buying and load sharing. This scheme includes two stages: in stage I, the two MNOs negotiate about their day-ahead aggregated energy purchase commitment by taking into account the real-time collaboration benefit in the future; while in stage II, the two MNOs negotiate about their real-time energy group buying as well as their wireless load sharing over time. Our proposed repeated Nash bargaining scheme is shown to achieve Pareto-optimal and fair energy cost reductions for both MNOs.
\end{itemize}

In the literature, there have been existing studies investigating the energy purchase of electricity consumers in the hybrid electricity market \cite{PhilpottOptimizing2006,VinyalsCoalitional2012}, in which the energy demands of consumers are assumed to be given. There have also been some works studying the load sharing in a single wireless network (see, e.g.,\cite{IsmailNetwork2011,BaoBayesianICC2015,GaoBargaining2014} and the references therein). In \cite{IsmailNetwork2011,BaoBayesianICC2015}, lightly-loaded BSs are allowed to offload their wireless traffic to co-located or nearby BSs with overlapped coverage, such that they can be turned to sleep mode for energy saving. In \cite{GaoBargaining2014}, load sharing between cellular BSs and WiFi access points (APs) is utilized to improve the payoffs of both the MNO (cellular BSs' owner) and the APs' owner.

In addition, there have been a handful of works \cite{BuWhen2012,HongPower2014,XuCost2015,ChiaEnergy2014,XuCoMP2015,GuoJoint2014,LeithonEnergy2014,XuCooperative2014} investigating the cellular networks powered by smart grids. In \cite{BuWhen2012,HongPower2014}, the authors developed joint optimization frameworks to maximize the utilities of both energy networks and cellular networks. In \cite{XuCost2015,ChiaEnergy2014,XuCoMP2015,GuoJoint2014,LeithonEnergy2014}, the authors considered an energy cooperation technique to improve the cost efficiency of cellular networks, in which distributed BSs are allowed to exchange their locally available (renewable) energy via leveraging two-way energy flows in smart grid infrastructures. In \cite{XuCooperative2014}, the authors considered the two-way energy trading between BSs and smart grids, in which the difference between energy buying and selling prices is exploited to reduce the energy cost at BSs. However, these prior works have not considered the hybrid electricity market that has been practically implemented in many countries, and have only focused on a single MNO case where the BSs are cooperative in nature and can be controlled by a centralizer. In contrast, this paper investigates the practical scenario with more than one self-interested MNOs in the hybrid electricity market, and provides incentive designs to motivate the inter-MNO collaboration. Furthermore, by borrowing insights from economics (e.g., \cite{ChenSegmenting2010}), this is also the first paper to leverage group buying to exploit the collaboration benefit between two MNOs for energy cost saving in the hybrid electricity market, which also requires wireless load sharing between different MNOs as demand management to fit the aggregated supply.

The remainder of this paper is organized as follows. Section \ref{sec:System} presents the hybrid electricity market and wireless traffic models. Section \ref{sec:Non} considers the non-cooperative case when the two MNOs operate independently to minimize their individual energy costs. Section \ref{sec:Fully} proposes the energy group buying with load sharing to minimize the total energy cost of the two MNOs when they are fully-cooperative. Section \ref{sec:Bargaining} develops a novel repeated Nash bargaining scheme for two self-interested MNOs to share their energy costs under energy group buying and load sharing. Section \ref{sec:Numerical} shows the simulation results. Finally, Section \ref{sec:Conclusion} concludes this paper.

\section{System Model}\label{sec:System}

\begin{figure}
\centering
 \epsfxsize=1\linewidth
    \includegraphics[width=8cm]{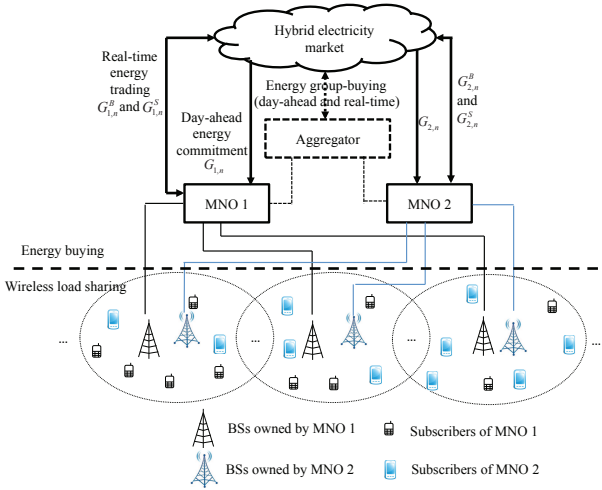}
\caption{System model with two MNOs coexisting in the same area, where each MNO purchases energy from the hybrid electricity market, and operates a number of BSs to serve subscribers.} \label{fig:0}
\end{figure}

\begin{figure}
\centering
 \epsfxsize=1\linewidth
    \includegraphics[width=8cm]{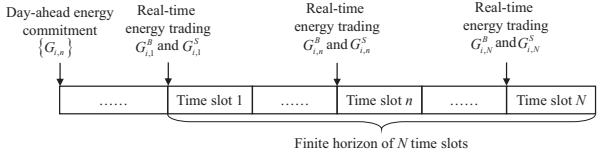}\vspace{-0em}
\caption{Day-ahead energy commitment and real-time energy trading of MNO $i$ over time.} \label{fig:time}\vspace{-0em}
\end{figure}

We consider two MNOs coexisting over the same geographic area (e.g., China Mobile and China Unicom in China) as shown in Fig. \ref{fig:0}. In the upper layer for energy buying in Fig. \ref{fig:0}, each MNO participates in the hybrid electricity market as a consumer. In the lower layer for wireless load sharing in Fig. \ref{fig:0}, each MNO operates a total of $K$ BSs within the area of interest to serve the respective subscribers. It is assumed that the $k$th BSs from the two MNOs (which are denoted by BSs $k_1$ and $k_2$ for MNOs 1 and 2, respectively) are geographically close or co-located to cover the same sub-area,{\footnote{In practice, different MNOs in China have agreed to jointly establish a telecommunications tower company, which will deploy shared tower infrastructures (cell sites including physical space, rooftops, towers, masts and pylons) to enable the co-location of individually deployed BSs by different MNOs \cite{reuters1,MeddourOn2011}. Note that our results can be extended to the case when the neighboring or paired BSs have {\it partially} overlapping coverage, by allowing more sophisticated wireless load sharing among them (e.g., each BS in one MNO may share traffic with multiple nearby BSs in the other MNO).}} $k\in\{1,\ldots,K\}$, and as a result, their wireless traffic can be shared/offloaded with each other. In practice, the pairing of BSs can be formed by the two MNOs to exchange the geographic locations of their BSs.

Due to our modeling of the hybrid electricity market including both the day-ahead and real-time energy purchase, we are interested in a finite horizon of each day consisting of $N>1$ time slots (each having a length of hours or tens of minutes). We consider quasi-static energy price and wireless traffic models, in which energy prices in the electricity market and wireless traffic at each BS remain constant within each time slot but may change from one slot to another. For notational convenience, we assume the length of each time slot is normalized to unity and thus we will use the terms ``energy'' and ``power'' interchangeably throughout this paper unless otherwise stated.

\subsection{Hybrid Electricity Market}

First, consider the upper layer for energy buying in Fig. \ref{fig:0}. In the hybrid electricity market, each MNO $i$ makes decisions on the day-ahead energy commitment and real-time energy trading over time as shown in Fig. \ref{fig:time}, for which the detailes are explained as follows.
\begin{itemize}
  \item {\it Energy purchase commitment in the day-ahead market}: In day-ahead (before the $N$-slot time horizon, e.g., before 4 pm of the previous day in the PJM day-ahead energy market \cite{OttExperience2003}), each MNO makes its energy purchase commitment for the following $N$ time slots based on its predicted energy demands. Let $G_{i,n}\ge 0$ denote the energy commitment by MNO $i\in\{1,2\}$ for time slot $n\in\{1,\ldots,N\}$, and $\alpha_n > 0$ denote the corresponding day-ahead energy price for time slot $n$. In practice, the energy price $\alpha_n$'s in general vary over different time slots (see Fig. \ref{fig:4}), and are announced by the power grid operator to the two MNOs in day-ahead.
  \item {\it  Energy trading in the real-time market}: The day-ahead energy commitments at each MNO may differ from its exact real-time energy demands. To overcome such mismatches, each MNO can buy the energy deficit or sell back the energy surplus in the real-time energy market. Let $G_{i,n}^{B}\ge 0$ and $G_{i,n}^{S}\ge 0$ denote the energy bought from or sold back to the real-time market by MNO $i$ in time slot $n$, and $\alpha^{B}_n \ge 0$ and $\alpha^{S}_n \ge 0$ denote the corresponding real-time energy buying and selling prices, respectively. In each time slot $n$, the real-time energy buying and selling prices are higher and lower than the day-ahead energy price, respectively, i.e., $\alpha^{S}_n <  \alpha_n < \alpha^{B}_n, \forall n\in\{1,\ldots,N\}$ (see Fig. \ref{fig:4}), such that electricity consumers will pay additional cost for the deviation of their real-time energy demands from the day-ahead commitments \cite{PhilpottOptimizing2006}. It is assumed that the real-time energy prices $\alpha^{B}_n$ and $\alpha^{S}_n$ are announced at the beginning of time slot $n$ in real-time but can only be partially predicted by MNOs in day-ahead with certain prediction errors. Let the predicted energy prices in  day-ahead be denoted as $\bar \alpha^{B}_n$ and $\bar \alpha^{S}_n$, respectively. Then we model the actual real-time energy buying and selling prices as
         \begin{align*}
         \alpha^{B}_n &= \bar \alpha^{B}_n + \delta^B_n,\\
         \alpha^{S}_n &= \bar \alpha^{S}_n + \delta^S_n,
        \end{align*}
        where $\delta^B_{n}$ and $\delta^S_n$ denote the prediction errors for the real-time energy buying and selling prices, both of which are modeled as random variables with zero mean and general probability density functions (PDFs) $\psi^B_{n}(\delta^B_{n})$ and $\psi^S_{n}(\delta^S_{n})$, respectively. $\delta^B_{n}$'s and $\delta^S_n$'s are assumed to be independent of each other.
\end{itemize}

\begin{figure}
\centering
 \epsfxsize=1\linewidth
    \includegraphics[width=7cm]{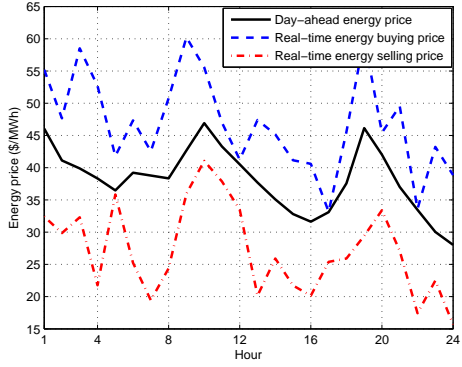}\vspace{-0em}
\caption{Day-ahead energy price $\{\alpha_n\}$ and real-time energy buying and selling prices $\{\alpha_n^B\}$ and $\{\alpha_n^S\}$ over one day, where $\{\alpha_n\}$ are set based on the practical day-ahead energy prices from PJM on March 01, 2015 \cite{PJM:data}.} \label{fig:4}\vspace{-0em}
\end{figure}

By combining the above day-ahead energy commitment and real-time energy trading, the exact energy purchased by MNO $i$ in time slot $n$ is given by $G_{i,n}+G_{i,n}^{B} - G_{i,n}^{S}$, and the corresponding energy cost of MNO $i$ in time slot $n$ is denoted as
\begin{align}
C_{i,n}(G_{i,n},G_{i,n}^B, G_{i,n}^S) &= \alpha_n G_{i,n}+ \alpha^{B}_n G^B_{i,n} - \alpha^{S}_n G^S_{i,n}, \nonumber \\
&~~ i\in\{1,2\},n\in\{1,\ldots,N\}.\label{eqn:cost1}
\end{align}

\begin{figure}
\centering
 \epsfxsize=1\linewidth
    \includegraphics[width=7cm]{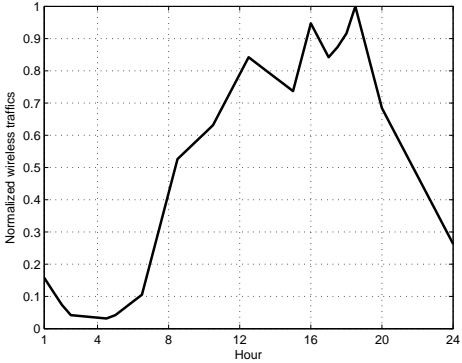}\vspace{-0em}
\caption{Normalized wireless traffic at a particular BS over one day, which is based on the measurement results from practical cellular BSs \cite{WillkommPrimary2009}.} \label{fig:3}\vspace{-0em}
\end{figure}

\subsection{Wireless Traffic Modeling and Sharing/Offloading at BSs}

Next, consider the lower layer for wireless load sharing in Fig. \ref{fig:0}. We first  introduce the wireless traffic models at cellular BSs. In time slot $n\in\{1,\ldots,N\}$, let $D_{k_i,n} > 0$ denote the actual wireless traffic requested by the subscribers of MNO $i\in\{1,2\}$ under the coverage sub-area of BS $k\in\{1,\ldots,K\}$. In general, $D_{k_i,n}$'s fluctuate considerably over time, as shown in Fig. \ref{fig:3}. It is assumed that the wireless traffic $D_{k_i,n}$'s, $\forall k\in\{1,\ldots,K\}, i\in\{1,2\}$, are perfectly known by MNO $i$ at the beginning of time slot $n$, but can only be partially predicted in day-ahead with certain prediction errors. Let $\bar D_{k_i,n} $ denote the predictable wireless traffic in day-ahead at BS $k$ of MNO $i$ (or BS $k_i$) for time slot $n$. Then we have
\begin{align}\label{eqn:D_error}
D_{k_i,n} = \bar D_{k_i,n} + \xi_{k_i,n},
\end{align}
where $\xi_{k_i,n} $ denotes the corresponding wireless traffic prediction error, which is modeled as a random variable with zero mean and a general PDF $\phi_{k_i,n}(\xi_{k_i,n})$. The distribution $\phi_{k_i,n}(\xi_{k_i,n})$'s are assumed to be known by the MNO, which can be practically obtained via gathering the historical data about the predicted and exact wireless traffic. Note that  $\xi_{k_i,n}$'s are assumed to be independent of each other and also independent of the energy price prediction error $\delta^S_n$'s and $\delta^B_{n}$'s.

Next, we present the wireless load sharing between the two MNOs. For each $k$th paired BSs from the two MNOs with the same covered sub-area (i.e., BSs $k_1$ and $k_2$), their wireless traffic can be shared/offloaded with each other (subject to maximum supportable wireless traffic as specified later), and after the load sharing, one of the two BSs (e.g., the lightly loaded one) may be turned into sleep mode. Note that the practical implementation of wireless load sharing between such paired BSs requires their subscribed mobile devices to be equipped with wireless interfaces used by both MNOs (e.g., one MNO uses new 5G techniques while the other MNO uses 4G LTE), which is reasonable for future smart phones and tablets.{\footnote{For example, nowadays iPhone 6 can support almost all cellular techniques including GSM/EDGE, CDMA EV-DO, UMTS/HSPA, TD-SCDMA, FDD-LTE, and TD-LTE \cite{AppleIphone6}.}} Also note that when the BS of one MNO is turned into sleep mode, the paired BS of the other MNO should ensure the coverage in the corresponding sub-area, where subscribers in the former MNO should be roamed to the latter MNO. Let the offloaded traffic from BS $k_i$ to BS $k_{\bar \imath}$ in time slot $n\in\{1,\ldots,N\}$ be denoted by $x_{k_i,n} \ge 0$, where $\bar \imath \in\{1,2\}\backslash \{i\}$ with $i\in\{1,2\}$. Then the wireless traffic served by BS $k_i$ in time slot $n$ is given by
\begin{align}\label{eqn:d_kin}
d_{k_i,n} = D_{k_i,n} - x_{k_i,n} + x_{k_{\bar \imath},n}.
\end{align}
Practically, due to the transmission power and bandwidth limitations, each cellular BS $k_i$ has a maximum supportable wireless traffic, which is denoted by $D^{\max}_{k_i} > 0$. As a result, we have
\begin{align}
d_{k_i,n} \le D^{\max}_{k_i}, \forall k\in\{1,\ldots,K\},i\in\{1,2\},n\in\{1,\ldots,N\}.
\end{align}
Note that the requested traffic load $D_{k_i,n}$ is assumed to be less than $D^{\max}_{k_i}$ by default, i.e., $D_{k_i,n} \le D^{\max}_{k_i},\forall k\in\{1,\ldots,K\},i\in\{1,2\},n\in\{1,\ldots,N\}$, since otherwise the service requests from subscribers would be denied due to the radio access control at BSs.

Based on the served wireless traffic, we then present the energy consumption models at BSs. For a typical cellular BS, when it is in the active mode, its energy consumption consists of two parts \cite{AuerHow2011}. The first part is the transmission energy consumption due to e.g. radio-frequency (RF) chains, which is related to the served wireless traffic at that BS. Specifically, we consider a linear transmission energy consumption model by defining the transmission energy consumed by BS $k_i$ as a product $a_{k_i}d_{k_i,n}$ in time slot $n$, where $a_{k_i} > 0$ denotes the rate of the energy consumption with respect to the traffic load. Here, the linear transmission energy consumption model is consistent with the measurement results in \cite{Peng}. The second part is the non-transmission energy consumption due to e.g. air conditioner and data processor, which is defined as a constant $b_{k_i}>0$ for BS $k_i$. On the other hand, when there is no wireless traffic served at the BS (i.e., $d_{k_i,n} = 0$), it can be turned into the sleep mode for power saving by switching off some of its components, for which the energy consumption is denoted by $c_{k_i}\ge 0$ with $c_{k_i} \ll b_{k_i}, \forall k\in\{1,\ldots,K\}, i\in\{1,2\}$. By combining the active and sleep modes, the energy consumption model for BS $k_i$ is expressed as
\begin{eqnarray}
P_{k_i}(d_{k_i,n})=&\left\{\begin{array}{ll} a_{k_i}d_{k_i,n} + b_{k_i}, & {\rm if}~
d_{k_i,n} >0 \\ c_{k_i}, &{\rm if}~d_{k_i,n} =0 \end{array} \right., \nonumber\\&~~~~~~~~~
k\in\{1,\ldots,K\}, i\in\{1,2\}.\label{eq1}
\end{eqnarray}
Here, the parameters $a_{k_i}$'s, $b_{k_i}$'s, and $c_{k_i}$'s can in general vary over BSs and MNOs because of different types of BSs and cellular techniques employed. Some typical values about these parameters can be referred to in \cite{ArnoldPowerConsumptionModeling2010}.

\subsection{Problem Statement}

Our objective is to minimize the sum energy cost of each MNO $i\in\{1,2\}$ over the $N$ time slots, i.e., $\sum_{n=1}^N C_{i,n}(G_{i,n},G_{i,n}^B, G_{i,n}^S)$ with $C_{i,n}(G_{i,n},G_{i,n}^B, G_{i,n}^S)$ given in (\ref{eqn:cost1}), while ensuring the wireless traffic requirements at BSs in each time slot $n$. Here, the wireless traffic requirement in time slot $n\in\{1,\ldots,N\}$ can be specified as the following constraint:
\begin{align}
G_{i,n}+G_{i,n}^{B} - G_{i,n}^{S} \ge \sum_{k=1}^{K} P_{k_i}(D_{k_i,n} - x_{k_i,n} + x_{k_{\bar \imath},n}),
\label{eqn:onstraints}
\end{align}
where the left-hand-side (LHS) denotes the purchased energy by MNO $i$ from the electricity market in time slot $n$, which should be no smaller than the total energy requirements at all $K$ BSs of MNO $i$, as given in the right-hand-side (RHS) based on (\ref{eqn:d_kin}) and (\ref{eq1}). The decision variables include the day-ahead energy commitment $\{G_{i,n}\}$, the real-time energy trading $\{G_{i,n}^{B}\}$ and $\{G_{i,n}^{S}\}$, and the wireless loading sharing $\{x_{k_i,n}\}$.

\section{No Cooperation Benchmark}\label{sec:NonCoop}\label{sec:Non}

In this section, we focus on the benchmark case when there is no cooperation between two traditional MNOs. In this case, the two MNOs operate independently without any wireless load sharing, i.e., $x_{k_i,n} = 0, \forall k\in\{1,\ldots,K\}, i\in\{1,2\}, n\in\{1,\ldots,N\}$. Then, from (\ref{eqn:d_kin}) and (\ref{eq1}), the served wireless traffic and energy demand of BS $k_i$ in time slot $n$ are $d_{k_i,n} = D_{k_i,n} > 0$ and $P_{k_i}(D_{k_i,n}) = a_{k_i}D_{k_i,n}+ b_{k_i}$, respectively.

We investigate the day-ahead and real-time energy purchase at each MNO to minimize its individual energy cost while ensuring the corresponding wireless traffic requirements. Without loss of generality, we focus on a particular MNO $i\in\{1,2\}$. As shown in Fig. \ref{fig:time}, MNO $i$ makes energy commitment $\{G_{i,n}\}$ in day-ahead (based on uncertain predictions of real-time energy prices and wireless traffic demands), and decides energy trading $G_{i,n}^{B}$ and $G_{i,n}^{S}$ at each time slot $n$ in real-time. Since the decisions are made at different time, we formulate the individual energy cost minimization as a two-stage stochastic programming problem. In stage I (day-ahead), MNO $i$ decides its day-ahead energy purchase commitment (i.e., $\{G_{i,n}\}$) to minimize the expected sum energy cost over the $N$ time slots, i.e.,
\begin{align}
\min\limits_{\{G_{i,n}\ge 0\}} & \mathbb{E}\bigg(\sum_{n=1}^NC_{i,n}^*(G_{i,n}) \bigg), \label{eqn:eq1}
\end{align}
where the expectation is over the real-time energy price prediction error $\delta^B_n$'s and $\delta^S_n$'s, and the wireless traffic prediction error $\xi_{k_i,n}$'s. Here, $C_{i,n}^*(G_{i,n})$ denotes the minimum energy cost under given $G_{i,n}$ in each time slot $n\in\{1,\ldots,N\}$ in stage II (real-time), which is achieved via optimizing the real-time energy trading (i.e., $G_{i,n}^{B}$ and $G_{i,n}^{S}$). That is,
\begin{align}
C_{i,n}^*(G_{i,n}) &= \min\limits_{G^B_{i,n}\ge 0,G^S_{i,n}\ge 0}~ C_{i,n}(G_{i,n},G_{i,n}^B, G_{i,n}^S) \nonumber\\
\mathrm{s.t.}~&G_{i,n}+G_{i,n}^{B}- G_{i,n}^{S} \ge \sum_{k=1}^{K} (a_{k_i}D_{k_i,n}+ b_{k_i}).\label{eqn:problem1}
\end{align}
In (\ref{eqn:problem1}), the constraint on the wireless traffic requirement follows from (\ref{eqn:onstraints}) together with $x_{k_i,n} = 0, \forall k\in\{1,\ldots,K\}, i\in\{1,2\}, n\in\{1,\ldots,N\}$. Note that in stage II, each MNO $i$ implements the real-time energy cost minimization in (\ref{eqn:problem1}) independently for each time slot $n$ (see Fig. \ref{fig:time}), since it can only acquire the exact real-time energy price $\alpha_n^B$ and $\alpha_n^S$ and wireless traffic demand $D_{k_i,n}$'s at the beginning of that time slot.

In the following, we employ the backward induction to solve this two-stage stochastic programming problem, by first obtaining the real-time energy trading solution in stage II (under any given $\{G_{i,n}\}$) and then deriving the day-ahead energy commitment solution in stage I.

\subsection{Real-Time Energy Trading in Stage II}\label{sec:Non:A}

First, we consider the real-time energy trading for MNO $i$ in stage II under any given day-ahead energy purchase commitment $G_{i,n}$'s. Since the MNO makes the real-time trading independently over different time slots, we solve the problem in (\ref{eqn:problem1}) by considering any one particular time slot $n\in\{1,\ldots,N\}$.

Problem (\ref{eqn:problem1}) is indeed a simple linear program (LP), for which we have the following lemma.

\begin{lemma}\label{lemma:1}
The optimal solution to problem (\ref{eqn:problem1}) under given $G_{i,n}$ is given by
\begin{align}
G^{B*}_{i,n}(G_{i,n}) &= \max\bigg(\sum_{k=1}^{K} (a_{k_i}D_{k_i,n}+ b_{k_i}) - G_{i,n}, 0\bigg) \label{eqn:solution1:1}\\
G^{S*}_{i,n}(G_{i,n}) &= \max\bigg(G_{i,n} - \sum_{k=1}^{K} (a_{k_i}D_{k_i,n}+ b_{k_i}), 0\bigg),\label{eqn:solution1:2}
\end{align}
which achieves the minimum real-time energy cost of MNO $i$ in time slot $n$ as
\begin{align}
C_{i,n}^*(G_{i,n}) = &
 \left\{ {\begin{array}{*{20}{c}}
\displaystyle {(\alpha_n-\alpha_n^B) G_{i,n} + \alpha_n^B\sum_{k=1}^{K} (a_{k_i}D_{k_i,n}+ b_{k_i}),}\\~~~~~~~~\displaystyle{{\rm if}~ G_{i,n} \le \sum_{k=1}^{K} (a_{k_i}D_{k_i,n}+ b_{k_i})}\\
\displaystyle
{(\alpha_n-\alpha_n^S) G_{i,n} + \alpha_n^S\sum_{k=1}^{K} (a_{k_i}D_{k_i,n}+ b_{k_i}),}\\~~~~~~~~\displaystyle{{\rm if}~G_{i,n} > \sum_{k=1}^{K} (a_{k_i}D_{k_i,n}+ b_{k_i})}.
\end{array}} \right.
\label{eqn:cost:optimal}
\end{align}
\end{lemma}
\begin{proof}
The optimal solution $G^{B*}_{i,n}(G_{i,n})$ and $G^{S*}_{i,n}(G_{i,n})$ in (\ref{eqn:solution1:1}) and (\ref{eqn:solution1:2}) can be easily obtained by solving the simple LP in (\ref{eqn:problem1}). Then, substituting (\ref{eqn:solution1:1}) and (\ref{eqn:solution1:2}) into the objective function of problem (\ref{eqn:problem1}), $C_{i,n}^*(G_{i,n})$ in (\ref{eqn:cost:optimal}) can be obtained.
\end{proof}

In Lemma \ref{lemma:1}, the optimal solution in (\ref{eqn:solution1:1}) and (\ref{eqn:solution1:2}) intuitively means that if the day-ahead energy commitment $G_{i,n}$ is more than the real-time energy demand $\sum_{k=1}^{K} (a_{k_i}D_{k_i,n}+ b_{k_i})$ of the $K$ BSs in the network, then MNO $i$ needs to buy the energy deficit from the real-time energy market; otherwise, MNO $i$ needs to sell back the excessive energy to the real-time energy market. It is evident from (\ref{eqn:solution1:1}) and (\ref{eqn:solution1:2}) that $G_{i,n}^{B*}(G_{i,n})$ and $G_{i,n}^{S*}(G_{i,n})$ cannot be positive at the same time, i.e., $G_{i,n}^{B*}(G_{i,n}) \cdot G_{i,n}^{S*}(G_{i,n}) = 0$. This means that it is not cost-effective for MNO $i$ to buy and sell energy at the same time in the real-time market. Furthermore, it can be shown that $C_{i,n}^*(G_{i,n})$ is a convex function with respect to the day-ahead energy commitment $G_{i,n}$ over $G_{i,n}\ge 0$, provided that $\alpha_n-\alpha_n^B < 0$ and $\alpha_n-\alpha_n^S>0$ (see an example of $C_{i,n}^*(G_{i,n})$ in Fig. \ref{fig:1}).

\begin{figure}
\centering
 \epsfxsize=1\linewidth
    \includegraphics[width=6cm]{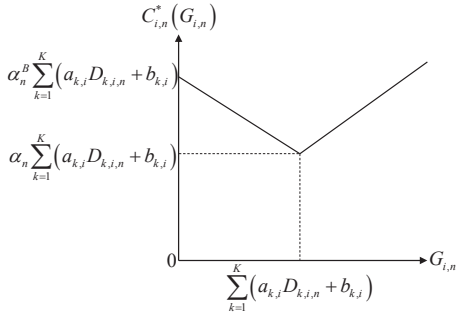}\vspace{-0em}
\caption{An example function $C_{i,n}^*(G_{i,n})$.} \label{fig:1}\vspace{-0em}
\end{figure}

\subsection{Day-Ahead Energy Commitment in Stage I}\label{sec:Non:B}

After obtaining the optimal real-time energy trading $\{G^{B*}_{i,n}(G_{i,n})\}$ and $\{G^{S*}_{i,n}(G_{i,n})\}$ in stage II, we then investigate stage I to obtain the optimal day-ahead energy commitment $\{G_{i,n}\}$, denoted by $\{G_{i,n}^*\}$, to solve problem (\ref{eqn:eq1}).

It is observed in problem (\ref{eqn:eq1}) that the variables $\{G_{i,n}\}$ are decoupling with each other, and as a result, problem (\ref{eqn:eq1}) can be decomposed into the following $N$ sub-problems, each for one time slot $n\in\{1,\ldots,N\}$.
\begin{align}
\min\limits_{G_{i,n}\ge 0} &\mathbb{E}\left(C_{i,n}^*(G_{i,n}) \right)\label{eqn:eq1:decompose}
\end{align}
Problem (\ref{eqn:eq1:decompose}) can be shown to be convex, since the objective function $\mathbb{E}\left(C_{i,n}^*(G_{i,n}) \right)$ (with expectation over the prediction error $\delta^B_n$'s, $\delta^S_n$'s, and $\xi_{k_i,n}$'s) is convex with respect to $G_{i,n}$ over $G_{i,n} \ge 0$, provided that $C_{i,n}^*(G_{i,n})$ is convex over $G_{i,n} \ge 0$ under any given $\delta^B_n$'s, $\delta^S_n$'s, and $\xi_{k_i,n}$'s \cite{convex}. However, standard convex optimization techniques cannot be directly applied to solve (\ref{eqn:eq1:decompose}), since the function $\mathbb{E}\left(C_{i,n}^*(G_{i,n}) \right)$ cannot be computed exactly in general. To overcome this issue, we employ the stochastic subgradient method \cite{Convex2} to solve this problem as follows by first presenting the exact subgradient of $\mathbb{E}\left(C_{i,n}^*(G_{i,n}) \right)$, then using the Monte Carlo method to approximate the subgradient, and finally applying the bisection method.

First, we show the exact subgradient of $\mathbb{E}\left(C_{i,n}^*(G_{i,n}) \right)$ by denoting the energy demand of MNO $i$ in time slot $n$ as $\zeta_{i,n} = \sum_{k=1}^{K} (a_{k_i}D_{k_i,n}+ b_{k_i})$, which is a random variable with the PDF given by $f_{i,n}(\zeta_{i,n})$.{\footnote{As defined in (\ref{eqn:D_error}), we have $D_{k_i,n} = \bar D_{k_i,n} + \xi_{k_i,n}$ with the prediction error $\xi_{k_i,n}$'s being random variables. As a result, $\zeta_{i,n}$ is also a random variable, whose PDF $f_{i,n}(\zeta_{i,n})$ can be obtained based on the PDFs of $\xi_{k_i,n}$'s (i.e., $\phi_{k_i,n}(\xi_{k_i,n})$'s).}} Then we have
\begin{align}
&\mathbb{E}\left(C_{i,n}^*(G_{i,n}) \right) =  (\alpha_n-\bar\alpha_n^S) G_{i,n}
\int_{G_{i,n}}^{\infty} f_{i,n}(\zeta_{i,n}) {\rm{d}}\zeta_{i,n} \nonumber\\&~~~~~~+ (\alpha_n-\bar\alpha_n^B) G_{i,n}  \int_0^{G_{i,n}} f_{i,n}(\zeta_{i,n}) {\rm{d}} \zeta_{i,n} \nonumber\\
&~~~~~~ + \bar \alpha_n^S\int_{G_{i,n}}^{\infty}  \zeta_{i,n} f_{i,n}(\zeta_{i,n}) {\rm{d}}\zeta_{i,n} \nonumber\\&~~~~~~ + \bar \alpha_n^B\int_{0}^{G_{i,n}} \zeta_{i,n} f_{i,n}(\zeta_{i,n}) {\rm{d}}\zeta_{i,n},\label{eqn:EC}
\end{align}
Note that $\mathbb{E}\left(C_{i,n}^*(G_{i,n}) \right)$ is only related to the predicted real-time energy price $\bar\alpha_n^S$ and $\bar\alpha_n^B$, but independent of the prediction error $\delta^S_{n}$'s and $\delta^B_{n}$'s, due to the fact that $\delta^S_{n}$'s and $\delta^B_{n}$'s are of zero mean. From (\ref{eqn:EC}), the exact subgradient of $\mathbb{E}\left(C_{i,n}^*(G_{i,n}) \right)$ can be obtained as
\begin{align}
g_{i,n}(G_{i,n}) =& (\alpha_n-\bar\alpha_n^S)\int_{G_{i,n}}^{\infty} f_{i,n}(\zeta_{i,n}) {\rm{d}}\zeta_{i,n} \nonumber\\&+ (\alpha_n-\bar\alpha_n^B) \int_0^{G_{i,n}} f_{i,n}(\zeta_{i,n}) {\rm{d}} \zeta_{i,n}.\label{eqn:subg}
\end{align}
Here, $g_{i,n}(G_{i,n})$ cannot be obtained directly due to the fact that the PDF $f_{i,n}(\zeta_{i,n})$ cannot be computed exactly in general.

Next, we use the Monte Carlo method to obtain an approximation of the subgradient $g_{i,n}(G_{i,n})$, denoted by $\hat g_{i,n}(G_{i,n})$. Specifically, based on the PDF $\phi_{k_i,n}(\xi_{k_i,n})$'s, we randomly generate $M$ independent samples of the wireless traffic prediction errors $\left\{\xi^{(m)}_{k_i,n}\right\}, m = 1,\ldots,M$, and accordingly compute the corresponding $M$ energy demands as $\zeta_{i,n}^{(m)} = \sum_{k=1}^{K} (a_{k_i}(\bar D_{k_i,n} + \xi_{k_i,n}^{(m)} )+ b_{k_i}), m=1, \ldots, M$. Then we have the approximate subgradient as
\begin{align*}
\hat g_{i,n}(G_{i,n}) = \frac{1}{M}&\bigg((\alpha_n-\bar\alpha_n^S)\sum_{m=1}^M\mv{1}_{\zeta_{i,n}^{(m)} > G_{i,n}} \bigg.\nonumber\\
&~~\bigg. + (\alpha_n-\bar\alpha_n^B) \sum_{m=1}^M\mv{1}_{\zeta_{i,n}^{(m)} \le G_{i,n}}\bigg),
\end{align*}
where $\mv{1}_A$ denotes the indicator function of an event $A$, with $\mv{1}_A = 1$ if $A$ is true, and $\mv{1}_A = 0$ otherwise. Evidently, when $M$ is sufficiently large, the approximate subgradient is a good estimate of the exact one, i.e., $\hat g_{i,n}(G_{i,n}) \approx  g_{i,n}(G_{i,n})$.

Finally, we use the bisection method \cite{convex} based on the approximate subgradient $\hat g_{i,n}(G_{i,n})$ to obtain the optimal day-ahead energy commitment solution $G_{i,n}^*$ to problem (\ref{eqn:eq1:decompose}) numerically. After acquiring the solution $G_{i,n}^*$ to (\ref{eqn:eq1:decompose}) for $n=1\ldots,N$, problem (\ref{eqn:eq1}) in stage I is solved.

To provide more insights on the optimal day-ahead energy trading, we have the following proposition by considering a special case of symmetric wireless traffic prediction errors (i.e., $\phi_{k_i,n}(\xi_{k_i,n}) = \phi_{k_i,n}(-\xi_{k_i,n}), \forall k\in\{1,\ldots,K\},i\in\{1,2\},n\in\{1,\ldots,N\}$).

\begin{proposition}\label{proposition:1}
For the optimal day-ahead energy commitment of MNO $i$ in time slot $n$, it follows that
\begin{itemize}
\item If $\alpha_n = ({\bar\alpha_n^B+\bar\alpha_n^S})/{2}$, then $G_{i,n}^* = \sum_{k=1}^{K} (a_{k_i}\bar D_{k_i,n}+ b_{k_i})$;
\item If $\alpha_n > ({\bar\alpha_n^B+\bar\alpha_n^S})/{2}$, then $G_{i,n}^* < \sum_{k=1}^{K} (a_{k_i}\bar D_{k_i,n}+ b_{k_i})$;
\item Otherwise, $G_{i,n}^* > \sum_{k=1}^{K} (a_{k_i}\bar D_{k_i,n}+ b_{k_i})$.
\end{itemize}
\end{proposition}
\begin{IEEEproof}
See Appendix  \ref{app:A}.
\end{IEEEproof}

Proposition \ref{proposition:1} intuitively shows that when the day-ahead energy price is larger than the average of the predicted real-time energy buying and selling prices, it is beneficial for the MNO to under-commit its predicted energy demand; while otherwise, it is desirable for the MNO to over-commit.

\section{Full Cooperation Through Energy Group Buying and Wireless Load Sharing}\label{sec:Fully}

In this section, we consider that the two MNOs are cooperative in nature (e.g., belonging to the same entity like Sprint and T-Mobile merged in some states of US \cite{HardySprint}), and fully cooperate in minimizing their total energy cost. Specifically, we propose a new approach named energy group buying associated with wireless load sharing, as described in detail as follows.

\subsection{Problem Formulation with Energy Group Buying}\label{sec:CooEneTrading}

Energy group buying is an approach that allows the two MNOs to share their day-ahead energy commitments and real-time energy trading. This approach is implemented with the assistance of an aggregator (see Fig. \ref{fig:0}), which is an entity commonly employed in smart girds to aggregate and control the demands at different electricity consumers (e.g., BSs of the two MNOs here) \cite{GkatzikisAggregator}. Via signing a contract with the aggregator, the two MNOs' BSs can be aggregated as a single group, and the aggregator can serve as an intermediary party to control the group of BSs to purchase energy from the hybrid electricity market. For convenience, under the energy group buying, we denote $G_{n} \ge 0$ as the aggregated day-ahead energy purchase commitment for the two MNOs in time slot $n$, and $G_{n}^B \ge 0$ and $G_{n}^S \ge 0$ as the aggregated real-time energy buying and selling amounts, respectively. Accordingly, the wireless traffic requirements at the two MNOs in (\ref{eqn:onstraints}) can be combined as
\begin{align}
G_{n}+G_{n}^{B} - G_{n}^{S} \ge \sum_{i=1}^2 \sum_{k=1}^{K} P_{k_i}\left(D_{k_i,n}-x_{k_i,n}+x_{k_{\bar \imath},n}\right),\label{eqn:15}
\end{align}
which means that the aggregated energy purchase of the two MNOs should be no smaller than the total energy demands of their BSs (with load sharing). Note that we also have the following energy group buying constraints for the two MNOs to share the day-ahead and real-time energy purchase.
\begin{align}
G_{1,n} + G_{2,n} &= G_{n}, \forall n\in\{1,\ldots,N\}\label{eqn:cooperative1}\\
G_{1,n}^B + G_{2,n}^B &= G_{n}^B, \forall n\in\{1,\ldots,N\}\label{eqn:cooperative2}\\
G_{1,n}^S + G_{2,n}^S &= G_{n}^S, \forall n\in\{1,\ldots,N\} \label{eqn:cooperative3}
\end{align}
Accordingly, the total energy cost of the two MNOs in time slot $n$ is expressed as
\begin{align}
C_{{\rm TC},n}(G_{n},G_{n}^B, G_{n}^S)& =\sum_{i=1}^2 C_{i,n}(G_{i,n},G_{i,n}^B, G_{i,n}^S) \nonumber\\& = \alpha_n G_{n}+ \alpha^{B}_n G^B_{n} - \alpha^{S}_n G^S_{n},\label{eqn:Ctotal2}
\end{align}
where $C_{i,n}(G_{i,n},G_{i,n}^B, G_{i,n}^S)$ is the energy cost of MNO $i$ in time slot $n$ as given in (\ref{eqn:cost1}).

Next, we consider the total energy cost minimization problem for the two MNOs. Similar to the non-cooperative case (see Fig. \ref{fig:time}), the day-ahead and real-time energy group buying decisions are made at different time. Therefore, we formulate the total energy cost minimization problem as a two-stage stochastic programming problem. In stage I (day-ahead), the two MNOs decide their day-ahead energy group buying (i.e., $\{G_{n}\}$) to minimize their expected sum energy cost over the $N$ time slots:{\footnote{We use the superscript $**$ to denote the optimal solution for the fully-cooperative case in this section, to distinguish from the superscript $*$ for the non-cooperative case in Section \ref{sec:Non}.}}
\begin{align}\label{eqn:day:ahead2}
\min_{\{G_{n}\ge 0\}}  \mathbb{E}\bigg(\sum_{n=1}^N C_{{\rm{TC}},n}^{**}(G_{n})\bigg).
\end{align}
Here, $C_{{\rm{TC}},n}^{**}(G_{n})$ denotes the minimum total energy cost under given $G_n$ in each time slot $n \in \{1,\ldots,N\}$ in stage II (real-time), which is achieved via optimizing the real-time energy group buying (i.e., $G_{n}^B$ and $G_{n}^S$) as well as the wireless load sharing (i.e., $\{x_{k_i,n}\}$), i.e.,
\begin{align}
C_{{\rm TC},n}^{**}(G_{n})  = &\min\limits_{G^B_{n},G^S_{n},\{x_{k_i,n}\}} C_{{\rm TC},n}(G_{n},G_{n}^B, G_{n}^S)  \label{eqn:problem2}\\
\mathrm{s.t.}~&
D_{k_i,n} - x_{k_i,n} + x_{k_{\bar \imath},n} \le D^{\max}_{k_i}, \nonumber\\
&~~~~~~~~~~\forall k\in\{1,\ldots,K\},i\in\{1,2\}\label{eqn:problem2:1}\\
&G^B_{n}\ge 0,G^S_{n} \ge 0, x_{k_i,n} \ge 0,\nonumber\\&~~~~~~~~~~ \forall k\in\{1,\ldots,K\},i\in\{1,2\}\label{eqn:problem2:3}\\
&(\ref{eqn:15}).\nonumber
\end{align}

It is worth noting that since we focus on the total energy cost in this section, only energy group buying decision variables $\{G_{n}\}$, $\{G_{n}^B\}$, and $\{G_{n}^S\}$ are of interest. In other words, under given energy group buying, any individual energy purchase variables $\{G_{i,n}\}$, $\{G_{i,n}^B\}$, and $\{G_{i,n}^S\}$ satisfying (\ref{eqn:cooperative1}), (\ref{eqn:cooperative2}), and (\ref{eqn:cooperative3}) are feasible to achieve the same total energy cost.

In the following, we solve the two-stage stochastic programming problem via the backward induction by first investigating stage II under any given day-ahead energy group buying $\{G_{n}\}$ and then studying stage I.

\subsection{Real-Time Energy Group Buying and Load Sharing in Stage II}

First, we consider the real-time energy group buying and wireless load sharing at the two MNOs in stage II under any given day-ahead energy group buying $\{G_n\}$. Since the real-time decisions are made from one slot to another independently, we focus on solving problem (\ref{eqn:problem2}) for a particular time slot $n\in\{1,\ldots,N\}$ without loss of generality.

Problem (\ref{eqn:problem2}) is not a convex optimization problem, since the constraint in (\ref{eqn:15}) is non-convex due to the non-convexity of the energy consumption function $P_{k_i}\left(\cdot\right)$'s in (\ref{eq1}). Despite this, we can still find the optimal solution to this problem by exploiting its specific structure. In the following, we first derive the optimal load sharing solution and then obtain the optimal real-time energy group buying solution.

\subsubsection{Optimal Load Sharing Solution}

It is observed that the objective of problem (\ref{eqn:problem2}) (i.e., the total energy cost of the two MNOs) is monotonically increasing as a function of the total energy demand of the two MNOs (i.e., $\sum_{i=1}^2 \sum_{k=1}^{K} P_{k_i}\left(D_{k_i,n}-x_{k_i,n}+x_{k_{\bar \imath},n}\right)$). As a result, the optimal load sharing solution can be obtained by minimizing the total energy demand of the two MNOs, i.e.,
\begin{align}
\min\limits_{\{x_{k_i,n}\ge 0\}} ~&\sum_{i=1}^2\sum_{k=1}^K P_{k_i}\left(D_{k_i,n}-x_{k_i,n}+x_{k_{\bar \imath},n}\right)\nonumber\\
\mathrm{s.t.}~~
&D_{k_i,n} - x_{k_i,n} + x_{k_{\bar \imath},n} \le D^{\max}_{k_i},~~\nonumber\\&~~~~~~~~~~~\forall k\in\{1,\ldots,K\},i\in\{1,2\}.\label{eqn:problem19:original}
\end{align}
Since $x_{k_i,n}$'s are decoupling over different BS pair $k$'s in both the objective function and constraints, problem (\ref{eqn:problem19:original}) can be decomposed into $K$ subproblems as follows, each corresponding to the total energy consumption minimization for the $k$th paired BSs, $k\in\{1,\ldots,K\}$.
\begin{align}
\min\limits_{\{x_{k_i,n}\ge 0\}} ~&\sum_{i=1}^2 P_{k_i}\left(D_{k_i,n}-x_{k_i,n}+x_{k_{\bar \imath},n}\right)\nonumber\\
\mathrm{s.t.}~~~
&D_{k_i,n} - x_{k_i,n} + x_{k_{\bar \imath},n} \le D^{\max}_{k_i},~~\forall i\in\{1,2\}\label{eqn:problem19}
\end{align}

By solving problem (\ref{eqn:problem19}), we have the following proposition.

\begin{proposition}\label{proposition:3}
The optimal load sharing solutions between the two MNOs' $k$th paired BSs, denoted by $x_{k_1,n}^{**}$ and $x_{k_2,n}^{**}$, $k\in\{1,\ldots,K\}$, are given as follows depending on the BSs' loads.
 \begin{itemize}
 \item {\it Case I}: any of the two BSs can support their total wireless traffic load, i.e., $D^{\max}_{k_i} \ge D_{k_1,n} + D_{k_2,n},\forall i\in\{1,2\}$. In this case, suppose that the total energy consumption with BS $k_i$ sleeping is no higher than that with BS $k_{\bar \imath}$ sleeping (i.e., $a_{k_{\bar\imath}} (D_{k_1,n} + D_{k_2,n}) + b_{k_{\bar\imath}} + c_{k_i} \le a_{k_i} (D_{k_1,n} + D_{k_2,n}) + b_{k_i} + c_{k_{\bar\imath}}$). Then it is optimal for BS $k_i$ to offload all its traffic to BS $k_{\bar \imath}$, i.e., $x_{k_i,n}^{**}=D_{k_i,n}$ and $x_{k_{\bar\imath},n}^{**}=0$, such that BS $k_i$ is turned into sleep mode.

\item {\it Case II}: only BS $k_{\bar\imath}$ can support their total wireless traffic load while the other BS $k_i$ cannot, i.e., $D^{\max}_{k_i} \le D_{k_1,n} + D_{k_2,n} \le D^{\max}_{k_{\bar \imath}}$. In this case, if the total energy consumption with BS $k_i$ sleeping is lower than that with both BSs active and BS $k_{\bar \imath}$ offloading the maximum traffic to BS $k_i$, i.e., $a_{k_{\bar\imath}} (D_{k_1,n} + D_{k_2,n}) + b_{k_{\bar\imath}} + c_{k_i} < a_{k_i} D^{\max}_{k_i} + b_{k_i} + a_{k_{\bar\imath}} (D_{k_1,n} + D_{k_2,n} - D^{\max}_{k_i}) + b_{k_{\bar\imath}}$, then it is optimal for BS $k_i$ to offload all its traffic to BS $k_{\bar \imath}$, i.e., $x_{k_i,n}^{**}= D_{k_i,n}$ and $x_{k_{\bar\imath},n}^{**}=0$, such that BS $k_i$ can be turned into sleep mode. Otherwise, it is optimal for BS $k_{\bar \imath}$ to offload the maximum traffic to BS $k_i$, i.e.,
    $x_{k_i,n}^{**}= 0$ and $x_{k_{\bar\imath},n}^{**}=D^{\max}_{k_{\bar\imath}} - D_{k_i,n}$, with both BSs active.

\item {\it Case III}: neither of the two BSs can support their total data traffic load, i.e., $D_{k_1,n} + D_{k_2,n} > D^{\max}_{k_i}, \forall \in\{1,2\}$. In this case, suppose that BS $k_i$ is no more energy efficient than BS $k_{\bar\imath}$, i.e., $a_{k_i} \ge a_{k_{\bar\imath}}$, then it is optimal for BS $k_i$ to offload the maximum traffic to BS $k_{\bar \imath}$, i.e., $x_{k_i,n}^{**}=D^{\max}_{k_{\bar\imath}} - D_{k_i,n}$ and $x_{k_{\bar\imath},n}^{**}=0$, with both BSs active.
\end{itemize}
\end{proposition}
\begin{IEEEproof}
See Appendix \ref{app:B}.
\end{IEEEproof}

\subsubsection{Optimal Real-Time Energy Group Buying Solution}

With the optimal load sharing obtained in Proposition \ref{proposition:3}, it remains to obtain the aggregated real-time energy buying and selling amounts $G_{n}^B$ and $G_n^S$ for the two MNOs. Substituting $\{x_{k_i,n}^{**}\}$ into problem (\ref{eqn:problem2}), then we have
\begin{align}
&C_{{\rm TC},n}^{**}(G_{n})  =
\min\limits_{G^B_{n}\ge 0,G^S_{n}\ge 0}~  \alpha_n G_{n}+ \alpha^{B}_n G^B_{n} - \alpha^{S}_n G^S_{n}  \nonumber\\
&~~\mathrm{s.t.}~ G_{n}+G_{n}^{B} - G_{n}^{S} \ge \sum_{i=1}^2\sum_{k=1}^{K} P_{k_i}\left(D_{k_i,n}-x_{k_i,n}^{**}+x^{**}_{k_{\bar \imath},n}\right).\label{eqn:problem2:2nd}
\end{align}
Problem (\ref{eqn:problem2:2nd}) has a similar structure as problem (\ref{eqn:problem1}), for which we have the following lemma.

\begin{lemma}\label{lemma:2}
The optimal solution to problem (\ref{eqn:problem2:2nd}) is given by
\begin{align}
&G_{n}^{B**}(G_{n}) \nonumber\\=& \max\bigg(\sum_{i=1}^2\sum_{k=1}^{K} P_{k_i}\left(D_{k_i,n}-x_{k_i,n}^{**}+x^{**}_{k_{\bar \imath},n}\right) - G_{n}, 0\bigg) \label{eqn:solution2:1}\\
&G_{n}^{S**}(G_{n}) \nonumber\\=& \max\bigg(G_{n} - \sum_{i=1}^2\sum_{k=1}^{K} P_{k_i}\left(D_{k_i,n}-x_{k_i,n}^{**}+x^{**}_{k_{\bar \imath},n}\right), 0\bigg),\label{eqn:solution2:2}
\end{align}
which achieves the minimum total energy cost of the two MNOs as
\begin{align}
&C_{{\rm{TC}},n}^{**}(G_{n}) = \nonumber\\&
 \left\{ {\begin{array}{*{20}{c}}
\displaystyle {(\alpha_n-\alpha_n^B) G_{n} + \alpha_n^B \sum_{i=1}^2\sum_{k=1}^{K} P_{k_i}\left(D_{k_i,n}-x_{k_i,n}^{**}+x^{**}_{k_{\bar \imath},n}\right),}\\~~~~~~~~ \displaystyle{{\rm if}~ G_{n}  \le \sum_{i=1}^2\sum_{k=1}^{K} P_{k_i}\left(D_{k_i,n}-x_{k_i,n}^{**}+x^{**}_{k_{\bar \imath},n}\right)}\\
\displaystyle
{(\alpha_n-\alpha_n^S) G_{n} + \alpha_n^S\sum_{i=1}^2\sum_{k=1}^{K} P_{k_i}\left(D_{k_i,n}-x_{k_i,n}^{**}+x^{**}_{k_{\bar \imath},n}\right),}\\~~~~~~~~\displaystyle{{\rm if}~G_{n} > \sum_{i=1}^2\sum_{k=1}^{K} P_{k_i}\left(D_{k_i,n}-x_{k_i,n}^{**}+x^{**}_{k_{\bar \imath},n}\right)}.
\end{array}} \right.
\label{eqn:cost:optimal:2}
\end{align}
Here, $C_{{\rm{TC}},n}^{**}(G_{n})$ is convex as a function of $G_{n}$ over $G_{n} \ge 0$.
\end{lemma}
\begin{proof}
This lemma can be similarly proved as Lemma \ref{lemma:1}, for which the proof is omitted for brevity.
\end{proof}

From Lemma \ref{lemma:2}, it is observed that the two fully-cooperative MNOs make the real-time energy group buying decision by comparing the aggregated day-ahead energy commitment versus the real-time total energy demand. This is similar to the real-time energy trading for non-cooperatively operated MNOs in Lemma \ref{lemma:1}, except that the total energy demand here is reshaped thanks to the wireless load sharing.

\subsection{Day-Ahead Energy Group Buying in Stage I}

After obtaining the optimal load sharing $\{x_{k_i,n}^{**}\}$ as well as the real-time energy group buying $G^{B**}_{n}(G_{n})$ and $G^{S**}_{n}(G_{n})$ in stage II, we then investigate stage I to obtain the optimal day-ahead energy group buying $\{G_{n}\}$ to solve problem (\ref{eqn:day:ahead2}).

Note that problem (\ref{eqn:day:ahead2}) has a very similar structure as problem (\ref{eqn:eq1}) in the non-cooperative case in Section \ref{sec:Non}. Therefore, problem (\ref{eqn:day:ahead2}) is indeed convex provided that $C_{{\rm{TC}},n}^{**}(G_{n})$ is a convex function in $G_{n} \ge 0$ under any given $\delta^B_n$'s, $\delta^S_n$'s, and $\xi_{k_i,n}$'s, as shown in Lemma \ref{lemma:2}. As a result, problem (\ref{eqn:day:ahead2}) can be similarly solved by the stochastic subgradient method as that for solving problem (\ref{eqn:eq1}) by replacing $\{G_{i,n}\}$ and $\left\{\zeta_{i,n} = \sum\limits_{k=1}^{K} (a_{k_i}D_{k_i,n}+ b_{k_i})\right\}$ in Section \ref{sec:Non:B} as  $\{G_{n}\}$ and $\left\{\zeta_{n} = \sum\limits_{i=1}^2\sum\limits_{k=1}^{K} P_{k_i}\left(D_{k_i,n}-x_{k_i,n}^{**}+x^{**}_{k_{\bar \imath},n}\right)\right\}$ here, respectively.{\footnote{Note that during computing the approximate subgradient of $\mathbb{E}\left(C_{{\rm TC},n}^{**}(G_{n})\right)$, we need to update the optimal load sharing $\{x_{k_i,n}^{**}\}$ for each of the $M$ randomly generated samples of wireless traffic prediction errors $\left\{\xi^{(m)}_{k_i,n}\right\}, m \in\{ 1,\ldots,M\}$.}} For brevity, we omit the detailed derivation of solving problem (\ref{eqn:day:ahead2}), and let its optimal solution be denoted by $\{G_{n}^{**}\}$. Therefore, the optimal day-ahead energy group buying solution has been obtained.

As compared to the optimal day-ahead energy commitment in the non-cooperative case in Section \ref{sec:Non}, the optimal day-ahead energy group buying here has taken into account the potential energy demand reduction due to the wireless load sharing, thus reducing the energy commitment and leading to significant total energy cost saving, as will be shown in the simulation results in Section \ref{sec:Numerical}.

\section{Repeated Nash Bargaining for Selfish MNOs' Cooperation}\label{sec:Bargaining}

In the previous section, we have considered the fully-cooperative case when the two MNOs minimize their total energy cost with the energy group buying and wireless load sharing implemented, where one MNO offloading more traffic may achieve large energy cost saving while the other MNO may be burdened with extra traffic. In practice, MNOs may belong to different entities and are self-interested in minimizing their individual energy costs. In this case, the above full cooperation with energy group buying and load sharing does not apply, and the two MNOs will reach the non-cooperative benchmark.

To overcome this issue, in this section we develop a novel incentive scheme, namely repeated Nash bargaining scheme,{\footnote{Please refer to \cite{NashTheBargaining1950} for an introduction about the two-person bargaining framework established by Nash.}} to motivate the collaboration between the two MNOs. In this scheme, the two MNOs negotiate about the energy group buying and load sharing, such that the energy cost reduction can be fairly shared between them. The negotiation between MNOs is assisted by the aggregator and implemented repeatedly both in day-ahead and at each time slot in real-time. Specifically, to maximize the flexibility of the MNOs' energy cost sharing and thus improve their cooperation incentive, we consider that the two MNOs implement monetary payments with each other at each time slot in real-time. Let the payment from MNO $i\in\{1,2\}$ to the other MNO $\bar\imath$ be denoted by $\pi_{i,n} \ge 0$ in time slot $n\in\{1,\ldots,N\}$.{\footnote{Note that the monetary payments have been commonly adopted in the literature to facilitate the revenue/payoff sharing among different parties (see e.g. \cite{GaoBargaining2014}). Here, the money exchange between the two MNOs can be performed daily by aggregating the payments over the whole day of $N$ time slots, thus minimizing the implementation complexity.}} Accordingly, the energy cost of MNO $i$ in time slot $n$ can be re-expressed as
\begin{align}
&\bar C_{i,n}(G_{i,n},G_{i,n}^B, G_{i,n}^S, \{\pi_{i,n}\}) = \alpha_n G_{i,n}+ \alpha^{B}_n G^B_{i,n} \nonumber\\&~~~~- \alpha^{S}_n G^S_{i,n} + \pi_{i,n} - \pi_{\bar\imath,n}, i\in\{1,2\},n\in\{1,\ldots,N\}.\label{eqn:barC}
\end{align}
With the monetary payments, we then describe the repeated bargaining between the two MNOs in detail as follows.
\begin{itemize}
\item {\it Stage I (day-ahead bargaining):} The two MNOs negotiate about the day-ahead energy group buying (i.e., $\{G_{n}\}$) and how to share the aggregated energy commitments between them (i.e., $\{G_{i,n}\}$ with $G_{1,n} + G_{2,n} = G_{n}, \forall n\in\{1,\ldots,N\}$ in (\ref{eqn:cooperative1})), by taking into account the potential real-time cooperation benefit. Suppose that the expected sum energy cost of MNO $i$ over the $N$ time slots is denoted by $\mathbb{E}\left(\sum\limits_{n=1}^N \bar C^{\star}_{i,n}(\{G_{i,n}\})\right)$, where $\bar C^{\star}_{i,n}(\{G_{i,n}\})$ represents the resulting energy cost of MNO $i$ in time slot $n$ under given $\{G_{i,n}\}$, which is based on the real-time bargaining at stage II as will be specified later (see (\ref{eqn:36})).{\footnote{We use the superscript $\star$ to denote the repeated Nash bargaining solution in this section, to distinguish from the superscript $*$ for the non-cooperative case in Section \ref{sec:Non} and $**$ for the fully-cooperative case in Section \ref{sec:Fully}.}} Then the payoff of MNO $i$ in day-ahead is defined as the expected energy cost reduction achieved by the energy group buying and wireless load sharing, i.e.,
    \begin{align}
    &U_i^{\rm Day}(\{G_{i,n}\}) = \mathbb{E}\bigg(\sum_{n=1}^NC_{i,n}^*(G^*_{i,n}) \bigg) \nonumber\\&~~~~- \mathbb{E}\bigg(\sum\limits_{n=1}^N \bar C^{\star}_{i,n}(\{G_{i,n}\})\bigg), i\in\{1,2\},\label{eqn:payoff:day}
    \end{align}
    where $\mathbb{E}\left(\sum_{n=1}^NC_{i,n}^*(G^*_{i,n}) \right)$ denotes the minimum expected sum energy cost of MNO $i$ achieved by the optimal non-cooperative energy purchase in Section \ref{sec:Non}.
\item {\it Stage II (real-time bargaining from time slot $1$ to $N$):} Under given day-ahead energy commitment (i.e., $\{G_{i,n}\}$) and  the corresponding day-ahead energy group buying (i.e., $\{G_{n} = G_{1,n} + G_{2,n}\}$), in real-time the two MNOs negotiate independently over each time slot $n\in\{1,\ldots,N\}$ about
    \begin{itemize}
    \item the real-time energy group buying (i.e., $G_{n}^B$ and $G_{n}^S$),
    \item how to share the aggregated energy trading amounts (i.e., $\{G_{i,n}^B\}$ and $\{G_{i,n}^S\}$ with $G^B_{1,n} + G^B_{2,n} = G^B_{n}$ and $G^S_{1,n} + G^S_{2,n} = G^S_{n}$ in (\ref{eqn:cooperative2}) and (\ref{eqn:cooperative3})),
    \item the wireless load sharing (i.e., $\{x_{k_i,n}\}$),
    \item and the inter-MNO payment (i.e., $\{\pi_{i,n}\}$).
    \end{itemize}
     The payoff of MNO $i$ in time slot $n$ is defined as the energy cost reduction, i.e.,
    \begin{align}
    &U_{i,n}^{\rm Real}(G_{i,n},G_{i,n}^B, G_{i,n}^S, \{\pi_{i,n}\}) \nonumber\\
    = &C_{i,n}^*(G_{i,n}) - \bar C_{i,n}(G_{i,n},G_{i,n}^B, G_{i,n}^S, \{\pi_{i,n}\}), i\in\{1,2\},
    \end{align}
    where $\bar C_{i,n}(G_{i,n},G_{i,n}^B, G_{i,n}^S, \{\pi_{i,n}\})$ denotes the energy cost of MNO $i$ with the real-time energy group buying and load sharing employed as given in (\ref{eqn:barC}), and $C_{i,n}^*(G_{i,n})$ denotes the minimum energy cost of MNO $i$ achieved by the optimal non-cooperative real-time energy trading in Section \ref{sec:Non:A}.
\end{itemize}

In the following, we use the backward induction together with the Nash bargaining solution \cite{NashTheBargaining1950} to address the above repeated bargaining problem, in which we first investigate the real-time bargaining in stage II under given day-ahead energy group buying, and then study the day-ahead bargaining in stage I.

\subsection{Real-Time Bargaining in Stage II}

First, we investigate the real-time bargaining between the two MNOs in stage II, under given day-ahead energy commitments $\{G_{i,n}\}$ and the according day-ahead energy group buying $\{G_{n} = G_{1,n} + G_{2,n}\}$. Since the MNOs negotiate with each other independently over different time slots, we consider a particular time slot $n\in\{1,\ldots,N\}$ without loss of generality.

In particular, each MNO $i$ aims to maximize its own energy cost reduction, or equivalently, the payoff $U_{i,n}^{\rm Real}(G_{i,n},G_{i,n}^B, G_{i,n}^S, \{\pi_{i,n}\})$. Such a negotiation can be formulated as a two-person bargaining problem \cite{NashTheBargaining1950}, in which MNO 1 and MNO 2 correspond to the two players of interest. The two MNOs can either reach an agreement on the real-time energy group buying and load sharing to achieve the payoffs $\left\{U_{i,n}^{\rm Real}(G_{i,n},G_{i,n}^B, G_{i,n}^S, \{\pi_{i,n}\})\right\}$, or fail to reach any agreements (i.e., operating as in the non-cooperative benchmark) to result in zero payoffs.

The reasonable bargaining solution has been characterized by Nash in \cite{NashTheBargaining1950} to achieve Pareto efficiency and symmetric fairness. Such a solution, termed the Nash bargaining solution, can be obtained by solving the following problem
\begin{align}
&\max_{G_{n}^B,G_{n}^S,\{G_{i,n}^B, G_{i,n}^S,x_{k_i,n}, \pi_{i,n}\}} ~
U_{1,n}^{\rm Real}(G_{1,n},G_{1,n}^B, G_{1,n}^S, \{\pi_{i,n}\}) \nonumber\\&~~~~~~~~~~~~~~~~~~~~~~~\cdot U_{2,n}^{\rm Real}(G_{2,n},G_{2,n}^B, G_{2,n}^S, \{\pi_{i,n}\}) \label{eqn:NBS}\\
&~~~\mathrm{s.t.}~ (\ref{eqn:cooperative2}),~(\ref{eqn:cooperative3}),~(\ref{eqn:15}),~(\ref{eqn:problem2:1}),~{\rm and}~(\ref{eqn:problem2:3})\nonumber\\
&~~~~~~~~U_{i,n}^{\rm Real}(G_{i,n},G_{i,n}^B, G_{i,n}^S, \{\pi_{i,n}\}) \ge 0, \forall i\in\{1,2\}\label{eqn:NBS:1}\\
&~~~~~~~~G_{i,n}^B\ge 0, G_{i,n}^S\ge 0,\pi_{i,n} \ge 0, \forall i\in\{1,2\}.\label{eqn:NBS:2}
\end{align}
Here, the constraints (\ref{eqn:cooperative2}) and (\ref{eqn:cooperative3}) specify the sharing of real-time energy group buying between MNOs. The constraints (\ref{eqn:15}), (\ref{eqn:problem2:1}), and (\ref{eqn:problem2:3}) specify the feasible sets of the real-time energy group buying $G_{n}^B$ and $G_{n}^S$ and load sharing $\{x_{k_i,n}\}$. The constraints in (\ref{eqn:NBS:1}) ensure that in the case of agreement, the payoff at each MNO is no smaller than that without any agreement.

Though problem (\ref{eqn:NBS}) is non-convex due to the non-convexity of the constraint (\ref{eqn:15}), it can still be optimally solved by exploiting the optimal solution to problem (\ref{eqn:problem2}) for the full cooperation. Let the optimal solution to problem (\ref{eqn:NBS}) be denoted by $G_{n}^{B\star},G_{n}^{S\star},\{G_{i,n}^{B\star}\}, \{G_{i,n}^{S\star}\},\{x^\star_{k_i,n}\}$, and $\{\pi^\star_{i,n}\}$. Then the following proposition follows.

\begin{proposition}\label{proposition:5.1}
The optimal real-time energy group buying and load sharing solutions to problem (\ref{eqn:NBS}) are given by $G_{n}^{B\star} = G^{B**}_{n}(G_{n})$, $G_{n}^{S\star} = G^{S**}_{n}(G_{n})$, and $x^\star_{k_i,n} = x^{**}_{k_i,n}, \forall k\in\{1,\ldots,K\},i\in\{1,2\}$, where $G^{B**}_{n}(G_{n})$, $G^{S**}_{n}(G_{n})$, and $x^{**}_{k_i,n}$ are the optimal solution to problem (\ref{eqn:problem2}) under given $G_n= G_{1,n} + G_{2,n}$. Substituting $G_{n}^{B\star},G_{n}^{S\star}$, and $\{x^\star_{k_i,n}\}$ into problem (\ref{eqn:NBS}), then the optimization of $\{G_{i,n}^{B\star}\}, \{G_{i,n}^{S\star}\}$, and $\{\pi^\star_{i,n}\}$ corresponds to a convex optimization problem, which can be solved via standard convex optimization techniques such as CVX \cite{CVX}. Accordingly, the resulting energy cost of MNO $i$ is given by
\begin{align}
&\bar C^{\star}_{i,n}(\{G_{i,n}\})  = \bar C_{i,n}(G_{i,n},G_{i,n}^{B\star}, G_{i,n}^{S\star}, \{\pi_{i,n}^\star\})  \nonumber\\&~~= \frac{C_{{\rm TC},n}^{**} (G_{n})}{2} + \frac{C_{i,n}^{*} (G_{i,n})-C_{\bar\imath,n}^{*} (G_{\bar\imath,n})}{2}, i\in\{1,2\}, \label{eqn:36}
\end{align}
where $C^{**}_{{\rm TC},n}(G_n)$ given in (\ref{eqn:cost:optimal:2}) is the minimum total energy cost of the two MNOs in the full cooperation under the day-ahead energy group buying $G_n$, and $C_{i,n}^{*} (G_{i,n})$ is the the minimum energy cost of MNO $i$ in the non-cooperative benchmark under the day-ahead energy commitment $G_{i,n}$.
\end{proposition}
\begin{IEEEproof}
See Appendix \ref{app:D}.
\end{IEEEproof}

\subsection{Day-Ahead Bargaining in Stage I}\label{sec:Bargaining:B}

After obtaining the real-time Nash bargaining solution in stage II, we next investigate the day-ahead bargaining between the two MNOs in stage I, where they negotiate about the day-ahead energy group buying (i.e., $\{G_{n}\}$) as well as the commitment sharing (i.e., $\{G_{i,n}\}$), so as to minimize their individual expected sum energy costs over the $N$ time slots. Such a negotiation is formulated as another two-person bargaining problem as follows.
\begin{itemize}
\item {\it Day-ahead energy group buying agreement.} When the two MNOs reach an agreement on the day-ahead energy group buying, they agree to make an aggregated energy commitment to minimize their expected total energy cost as in the case of full cooperation, i.e., $G_n^{\star} = G_{n}^{**}, \forall n\in\{1,\ldots,N\}$ and accordingly
    \begin{align*}
    G_{1,n} + G_{2,n} = G_{n}^{**}, \forall n\in\{1,\ldots,N\}.
    \end{align*}
    In this case, by substituting (\ref{eqn:36}) into (\ref{eqn:payoff:day}) and replacing $\{G_{n}\}$ as $\{G_{n}^{**}\}$, the payoff of MNO $i\in\{1,2\}$ in time slot $n\in\{1,\ldots,N\}$ can be re-expressed as
    \begin{align}
    &U_i^{\rm Day}(\{G_{i,n}\}) \nonumber\\ = &\mathbb{E}\bigg(\sum_{n=1}^NC_{i,n}^*(G^*_{i,n}) \bigg) \nonumber\\&- \mathbb{E}\bigg(\sum\limits_{n=1}^N \left(\frac{C_{{\rm TC},n}^{**} (G^{**}_{n})}{2} + \frac{C_{i,n}^{*} (G_{i,n})-C_{\bar\imath,n}^{*} (G_{\bar\imath,n})}{2}\right)\bigg) \nonumber\\
     = &\Upsilon_i + \frac{1}{2}\sum\limits_{n=1}^N \bigg(\mathbb{E}\big(C_{\bar\imath,n}^{*} (G_{\bar\imath,n})\big) - \mathbb{E}\big(C_{i,n}^{*} (G_{i,n})\big)\bigg),\label{eqn:payoff:day:2}
    \end{align}
    where $\displaystyle\Upsilon_i = \mathbb{E}\bigg(\sum\limits_{n=1}^NC_{i,n}^*(G^*_{i,n}) \bigg) - \frac{1}{2}\mathbb{E}\bigg(\sum\limits_{n=1}^N C_{{\rm TC},n}^{**} (G^{**}_{n})\bigg)$ is a constant.
\item {\it Day-ahead disagreement.} When the two MNOs fail to reach any agreement in day-ahead, they will operate independently as in the non-cooperative case in both day-ahead and real-time. In this case, both MNOs achieve zero payoffs.
\end{itemize}

For the day-ahead negotiation, the desirable Nash bargaining solution can be obtained by solving the following optimization problem:
\begin{align}
\max\limits_{\{G_{i,n}\}} ~& U_1^{\rm Day}(\{G_{i,n}\})\cdot U_2^{\rm Day}(\{G_{i,n}\})\label{eqn:38}\\
\mathrm{s.t.}~~&U_i^{\rm Day}(\{G_{i,n}\}) \ge 0,\forall i\in\{1,2\}\label{eqn:38:1}\\
& G_{1,n} + G_{2,n} = G_{n}^{**}, \forall n\in\{1,\ldots,N\}.\label{eqn:38:2}
\end{align}
where (\ref{eqn:38:1}) ensures that in the case of agreement, the payoff at each MNO is no smaller than that without any agreements, and (\ref{eqn:38:2}) specifies the MNOs' energy commitment sharing of the day-ahead group buying.

Note that the function $C_{i,n}^{*} (G_{i,n})$'s (and thus $\mathbb{E}\left(C_{i,n}^{*} (G_{i,n})\right)$'s), $\forall i\in\{1,2\},n\in\{1,\ldots,N\}$, are all convex functions (see (\ref{eqn:cost:optimal})), and as a result, problem (\ref{eqn:38}) is non-convex in general. This, together with the fact that $\mathbb{E}\left(C_{i,n}^{*} (G_{i,n})\right)$ cannot be computed exactly, makes problem (\ref{eqn:38}) very difficult to be solved optimally. To overcome this issue, we make the following approximation on $\mathbb{E}\left(C_{i,n}^{*} (G_{i,n})\right)$:
\begin{align}
&\mathbb{E}\left(C_{i,n}^{*} (G_{i,n})\right)  \nonumber\\\approx &
 \left\{ {\begin{array}{*{20}{c}}
\displaystyle {(\alpha_n-\alpha_n^B) G_{i,n} + \alpha_n^B\sum_{k=1}^{K} (a_{k_i} \bar D_{k_i,n}+ b_{k_i}),}\\~~~~~~~~~~~~~~~\displaystyle{{\rm if}~ G_{i,n} \le \sum_{k=1}^{K} (a_{k_i}\bar D_{k_i,n}+ b_{k_i})}\\
\displaystyle
{(\alpha_n-\alpha_n^S) G_{i,n} + \alpha_n^S\sum_{k=1}^{K} (a_{k_i} \bar D_{k_i,n}+ b_{k_i}),}\\~~~~~~~~~~~~~~~\displaystyle{{\rm if}~G_{i,n} > \sum_{k=1}^{K} (a_{k_i}\bar D_{k_i,n}+ b_{k_i})}
\end{array}} \right.
\label{eqn:cost:optimal:app1}\\
\approx & {(\alpha_n-\alpha_n^B) G_{i,n} + \alpha_n^B\sum_{k=1}^{K} (a_{k_i} \bar D_{k_i,n}+ b_{k_i})}.\label{eqn:cost:optimal:app2}
\end{align}
where in (\ref{eqn:cost:optimal:app1}) the predicted wireless traffic $\bar D_{k_i,n}$'s are treated as the exact ones by ignoring the prediction errors, and in (\ref{eqn:cost:optimal:app2}) the case of $G_{i,n} > \sum_{k=1}^{K} (a_{k_i}\bar D_{k_i,n}+ b_{k_i})$ is omitted based on the observation that the day-ahead energy commitment $G_{i,n}$ with load sharing is highly likely to be smaller than the predicted energy demand $\sum_{k=1}^{K} (a_{k_i}\bar D_{k_i,n}+ b_{k_i})$ without load sharing.

Using the approximation in (\ref{eqn:cost:optimal:app2}) to replace $\mathbb{E}\left(C_{i,n}^{*} (G_{i,n})\right), \forall k\in\{1,\ldots,K\},i\in\{1,2\}$, problem (\ref{eqn:38}) then becomes a convex optimization problem, which can thus be solved by standard convex optimization techniques such as CVX \cite{CVX}. Let the solution obtained accordingly at the two MNOs be denoted by $\{G_{i,n}^\star\}$ for the $N$ time slots. By combining $\{G_{i,n}^\star\}$ and $\{G_{n}^\star\}$, the day-ahead bargaining solution has been finally obtained.

Despite that a sub-optimal approach (with the approximation in (\ref{eqn:cost:optimal:app2})) has been used to solve problem (\ref{eqn:38}), we have the following proposition.

\begin{proposition}
Our day-ahead bargaining solution can achieve Pareto optimal energy costs at the two MNOs.
\end{proposition}
\begin{proof}
This proposition holds due to the fact that by setting the day-ahead energy group buying as $G_{n}^\star = G_{n}^{**},\forall n\in\{1,\ldots,N\}$, the day-ahead bargaining solution can always obtain the minimum expected total energy cost for the two MNOs.
\end{proof}

Besides the Pareto optimality, our day-ahead bargaining solution can also ensure the fairness of energy cost reduction between the two MNOs, as will be shown in the simulation results next.

\section{Simulation Results}\label{sec:Numerical}

In this section, we provide simulation results to compare the performances of the non-cooperative benchmark, the full cooperation, and the repeated Nash bargaining scheme in terms of the MNOs' energy costs. We consider that there are $K = 500$ LTE macro BSs deployed by each of the two MNOs, and the one day horizon of our interest consists of $N=48$ time slots each with a length of 30 minutes. For LTE macro BSs, the maximum average supportable data rate throughput for each BS is set to be $D_{k_i}^{\max} =D^{\max}= 150$ Mbps,{\footnote{The average downlink data rate throughput in one LTE cell is reported to be about 50 Mbps with a bandwidth of 20 MHz and $4\times 4$ multiple-input multiple-output (MIMO) employed. Assuming there are three sectors at each BS, we use $D^{\max} = 150$ Mbps for the convenience of analysis.}} and the power consumption parameters are set as $a_{k_i} = 12$ Watt (W)/Mbps, 
 $b_{k_i} = 1200$ W, and $c_{k_i} = 30$ W, $\forall k\in\{1,\ldots,K\}, i\in\{1,2\}$ \cite{ArnoldPowerConsumptionModeling2010}. 
Furthermore, we use the normalized wireless traffic profiles over one day in Fig. \ref{fig:3}, denoted by $\{\theta_n\}$, as a reference to generate the wireless traffic for the BSs in our simulations. Based on $\{\theta_n\}$, we generate the real-time wireless traffic at BS $k_i$ as $D_{k_i,n} = \chi_{k_i} \theta_n, \forall n\in\{1,\ldots,N\}$, where the amplitude $\chi_{k_i}$ is a random variable as will be specified later. We randomly generate the wireless traffic prediction error $\xi_{k_i,n}$ following a uniform distribution over the interval $[-0.4 D_{k_i,n},0.4 D_{k_i,n}]$. As for energy prices, we use the day-ahead energy price $\{\alpha_n\}$ as well as real-time energy buying and selling prices $\{\alpha_n^B\}$ and $\{\alpha_n^S\}$ over one day as shown in Fig. \ref{fig:4}. The energy price prediction errors $\delta_n^B$ and $\delta_n^S$ are set to be random variables following uniform distributions over the intervals $[-0.1\alpha_n^B, 0.1\alpha_n^B]$ and $[-0.1\alpha_n^S, 0.1\alpha_n^S]$, respectively. Note that the specific wireless traffic and energy price prediction errors are set here for simulation, while consistent results can be obtained under other parameter values. We conduct 100 independent realizations to obtain the results in this section.

\subsection{The Case with Symmetric Wireless Traffic}

In Figs. \ref{fig:5} and \ref{fig:6}, we consider that the two MNOs have symmetric wireless traffic, in which the amplitude of traffic pattern $\chi_{k_i}$'s, $\forall k\in\{1,\ldots,K\}, i\in\{1,2\}$, are generated as uniformly distributed random variables over the continuous interval $[0.1D^{\max}, 0.9D^{\max}]$. Fig. \ref{fig:5} shows the total energy cost of the two MNOs under the full cooperation as compared to that under the non-cooperative benchmark. It is observed that during the lightly-loaded periods (e.g., hours 1-8), more significantly total energy cost reductions are achieved as compared to the heavily-loaded periods (e.g., hours 12-18). This is due to the fact that during lightly-loaded periods, it is more likely that activating one BS is sufficient to satisfy the total wireless traffic demands over each sub-area, and thus the other BS can be turned into sleep mode to save power. It is calculated that over the whole day of $N=48$ time slots, the full cooperation achieves an average of 22.81\% energy cost reduction as compared to the non-cooperative benchmark.

\begin{figure}
\centering
 \epsfxsize=1\linewidth
    \includegraphics[width=7cm]{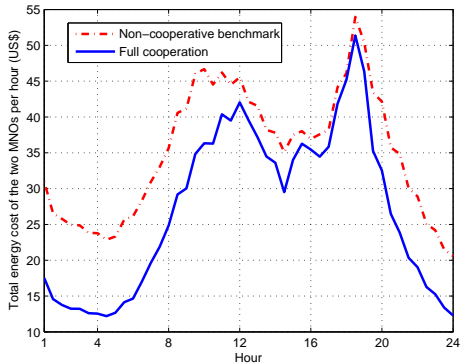}\vspace{-0em}
\caption{The average per-hour total energy cost of the two MNOs over one day under symmetric wireless traffic.} \label{fig:5}\vspace{-0em}
\end{figure}

\begin{figure}
\centering
 \epsfxsize=1\linewidth
    \includegraphics[width=7cm]{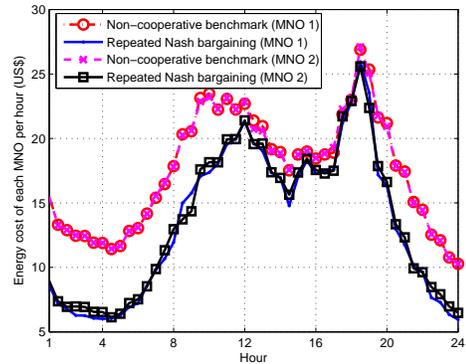}\vspace{-0em}
\caption{The average per-hour energy cost at each of the two MNOs over one day under symmetric wireless traffic.} \label{fig:6}\vspace{-0em}
\end{figure}

Fig. \ref{fig:6} shows the individual energy costs at the two MNOs under the non-cooperative benchmark and the repeated Nash bargaining scheme, respectively. Note that the repeated Nash bargaining scheme achieves the same total energy costs for the two MNOs as the full cooperation, and thus we have not shown the full cooperation in this figure. It is observed that under both the non-cooperative benchmark and the repeated Nash bargaining scheme, the achieved energy costs at the two MNOs have very similar patterns, given the fact that their wireless traffic demands are symmetric. Over the whole day of $N=48$ time slots, the sum energy costs at the two MNOs under the repeated Nash bargaining are calculated to be US\$ 648.88 and  US\$ 646.53, respectively. As compared to the sum energy costs US\$ 839.35 and US\$ 838.86 under the non-cooperative benchmark, 22.69\% and 22.92\% energy cost reductions are achieved by the two MNOs, respectively.

\begin{figure}
\centering
 \epsfxsize=1\linewidth
    \includegraphics[width=7cm]{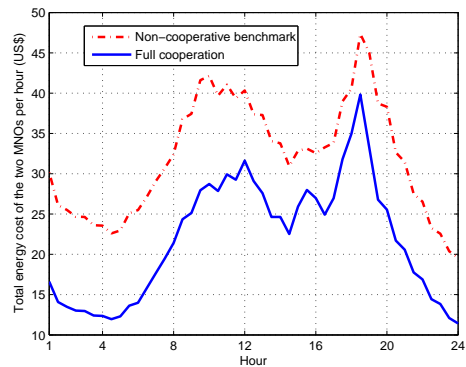}\vspace{-0em}
\caption{The average per-hour total energy cost of the two MNOs over one day under asymmetric wireless traffic.} \label{fig:7}\vspace{-0em}
\end{figure}

\subsection{The Case with Asymmetric Wireless Traffic}

In Figs. \ref{fig:7} and \ref{fig:8}, we consider that the wireless traffic demands at the two MNOs are asymmetric, i.e., the wireless data traffic at MNO 2 is lighter than that at MNO 1. Specifically, we generate the amplitude of traffic pattern ${\chi}_{k_1}$'s for MNO 1, $\forall k\in\{1,\ldots,K\}$, as uniformly distributed random variables over the continuous interval $[0.1D^{\max}, 0.9D^{\max}]$, and $\chi_{k_2}$'s for MNO 2 as uniformly distributed random variables over $[0.05D^{\max}, 0.45D^{\max}]$. Note that each of the two MNOs is assumed to have the same number of $K=500$ BSs, which may correspond to the scenario when they share the same cellular infrastructures (cell sites) \cite{reuters1}. When the two MNOs deploy BSs individually, the MNO with lower traffic load may deploy less BSs with increased coverage. In this case, the implementation of the energy group buying with wireless load sharing requires more sophisticated inter-MNO traffic offloading (e.g., each BS in the lightly loaded MNO may share traffic with multiple neighboring BSs in the other heavily loaded MNO), which, however, is beyond the scope of this paper and will be left for our future work.

Fig. \ref{fig:7} shows the total energy cost of the two MNOs under the non-cooperative benchmark and the full cooperation, respectively. It is observed that full cooperation achieves considerable total cost reduction over the whole day, even during the relatively heavily-loaded hours 12-18. This is due to the fact that MNO 2 is more lightly-loaded and thus its BSs are highly likely to be turned into sleep mode even during the hours 12-18. Over the whole day, 31.72\% total energy cost reduction is achieved by the full cooperation as compared to the non-cooperative benchmark.

\begin{figure}
\centering
 \epsfxsize=1\linewidth
    \includegraphics[width=7cm]{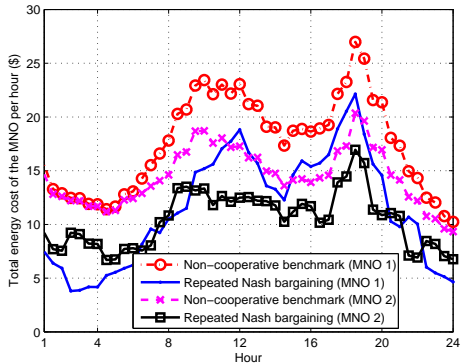}\vspace{-0em}
\caption{The average per-hour energy cost at each of the two MNOs over one day under asymmetric wireless traffic.} \label{fig:8}\vspace{-0em}
\end{figure}

Fig. \ref{fig:8} shows the individual energy costs at the two MNOs under the non-cooperative benchmark and the repeated Nash bargaining scheme, respectively. Under the non-cooperative benchmark, it is observed that MNO 2 has lower energy costs than MNO 1 over all the $48$ time slots, which is consistent with their wireless traffic demands. In contrast, under the repeated Nash bargaining scheme, it is observed that over the $48$ time slots, the energy cost patterns of MNO 1 and MNO 2 fluctuate considerably. For example, MNO 1 achieves more significant energy cost reductions than MNO 2 during hours 1-4, whereas the opposite is true during hour 12. The reason is that during the day-ahead bargaining in stage I (see Section \ref{sec:Bargaining:B}), the two MNOs negotiate to minimize their sum energy costs over $N$ time slots as a whole, and therefore, their energy costs at each time slot may fluctuate. Over the whole day, the sum energy costs at the two MNOs are computed to be US\$ 548.13 and US\$ 500.05 under the repeated Nash bargaining scheme, versus US\$ 840.47 and US\$ 694.75 under the non-cooperative benchmark. Accordingly, 34.78\% and 28.02\% energy cost reductions are achieved by the two MNOs, respectively.

\section{Concluding Remarks}\label{sec:Conclusion}

In this paper, we considered cellular networks in the emerging hybrid electricity market, and proposed a new approach named energy group buying with load sharing to reduce the energy costs of MNOs via exploiting their collaboration benefit. To motivate self-interested MNOs to cooperate, we further developed a novel repeated Nash bargaining scheme, in which the MNOs negotiate about the energy group buying and load sharing in both day-ahead and real-time, so as to fairly share the energy cost reduction. The repeated Nash bargaining scheme is shown to achieve the Pareto optimal and fair energy cost reductions for the MNOs. With the emerging 5G cellular and smart grid systems, it is our hope that this paper can open a new avenue towards how to jointly utilize new communication techniques (e.g., wireless load sharing) and smart grid techniques (e.g., energy group buying) to further improve the energy cost efficiency for future green cellular networks.

Note that our results in this paper are extendable to the scenario with more than two MNOs. When these MNOs belong to the same entity, they should be grouped all together for energy group buying and wireless loading sharing. How to decide the BSs' active/sleep modes and the wireless load sharing among them in each sub-area is a challenging task. On the other hand, when these MNOs belong to different entities, they should further engage in group bargaining \cite{ChaeAGroup} (instead of one-to-one bargaining in this paper) to fairly share the potential collaboration benefit achieved by energy group buying with wireless loading sharing. It is important but challenging to design the bargaining protocol among these MNOs (e.g., sequential or concurrent bargaining \cite{GaoBargaining2014}), depending on individual bargaining power.

It is also worth pointing out that our proposed energy group buying with load sharing can be implemented in emerging wireless heterogeneous networks (HetNets) with densely deployed small cell BSs such as pico- and femto-BSs, by carefully pairing them via taking into account their diverse energy consumption, coverage, and communication capabilities. In particular, wireless load sharing between different types of BSs (e.g., between macro and small cell BSs) can help further improve the energy cost efficiency of HetNets. Since  small cell BSs are often located within the covered sub-areas of macro BSs and need not take care of the coverage, paired small cell BSs from different MNOs can be  turned into sleep mode at the same time by offloading their traffic to macro BSs. How to optimize the inter-MNO wireless load sharing among different BS types together with the energy group buying is an interesting issue worth pursuing in future work.

\appendix
\subsection{Proof of Proposition \ref{proposition:1}}\label{app:A}

Note that $\zeta_{i,n}= \sum_{k=1}^{K} (a_{k_i}\bar D_{k_i,n}+ b_{k_i}) + \Xi_{i,n}$ with $\Xi_{i,n} = \sum_{k=1}^{K} (a_{k_i}\xi_{k_i,n})$. Therefore, given the fact that $\phi_{k_i,n}(\xi_{k_i,n}) = \phi_{k_i,n}(-\xi_{k_i,n})$, it follows that $f_{i,n}(\sum_{k=1}^{K} (a_{k_i}\bar D_{k_i,n}+ b_{k_i}) + \Xi_{i,n}) = f_{i,n}(\sum_{k=1}^{K} (a_{k_i}\bar D_{k_i,n}+ b_{k_i}) -  \Xi_{i,n})$. Accordingly, it follows that
\begin{align}
&\int_{\sum_{k=1}^{K} (a_{k_i}\bar D_{k_i,n}+ b_{k_i}) }^{\infty} f_{i,n}(\zeta_{i,n}) {\rm{d}}\zeta_{i,n} \nonumber\\=&\int_0^{\sum_{k=1}^{K} (a_{k_i}\bar D_{k_i,n}+ b_{k_i}) } f_{i,n}(\zeta_{i,n}) {\rm{d}} \zeta_{i,n}.\label{eqn:43}
\end{align}
From (\ref{eqn:43}) and (\ref{eqn:subg}), we have
\begin{align}
&g_{i,n}\left(\sum_{k=1}^{K} (a_{k_i}\bar D_{k_i,n}+ b_{k_i})\right)  \nonumber\\=& (2\alpha_n-\bar\alpha_n^S-\bar\alpha_n^B) \int_{\sum_{k=1}^{K} (a_{k_i}\bar D_{k_i,n}+ b_{k_i}) }^{\infty} f_{i,n}(\zeta_{i,n}) {\rm{d}}\zeta_{i,n}.\label{eqn:g:G}
\end{align}

From (\ref{eqn:g:G}), it follows that when $\alpha_n = ({\alpha_n^B+\alpha_n^S})/{2}$, we have $g_{i,n}\left(\sum_{k=1}^{K} (a_{k_i}\bar D_{k_i,n}+ b_{k_i})\right) = 0$, and accordingly the optimal solution is given as $G_{i,n}^* = \sum_{k=1}^{K} (a_{k_i}\bar D_{k_i,n}+ b_{k_i})$. Similarly, if $\alpha_n > ({\alpha_n^B+\alpha_n^S})/{2}$ or $\alpha_n < ({\alpha_n^B+\alpha_n^S})/{2}$, then it follows that $g_{i,n}\left(\sum_{k=1}^{K} (a_{k_i}\bar D_{k_i,n}+ b_{k_i})\right) > 0$ or $g_{i,n}\left(\sum_{k=1}^{K} (a_{k_i}\bar D_{k_i,n}+ b_{k_i})\right) < 0$. As a result, we have $G_{i,n}^* < \sum_{k=1}^{K} (a_{k_i}\bar D_{k_i,n}+ b_{k_i})$ or $G_{i,n}^* >\sum_{k=1}^{K} (a_{k_i}\bar D_{k_i,n}+ b_{k_i})$, respectively. Therefore, this proposition is proved.

\subsection{Proof of Proposition \ref{proposition:3}}\label{app:B}

We prove this proposition by considering the three cases one by one.

First, consider Case I when $D^{\max}_{k_i} \ge D_{k_1,n} + D_{k_2,n},\forall i\in\{1,2\}$. In this case, it can be shown that it is optimal for the two BSs to offload their total traffic loads into one of them such that the other can be turned into sleep mode to save energy, since otherwise higher non-transmission power is consumed. Based on this observation, the optimal solution in this case can be obtained by comparing the energy consumption under the two possible solutions. The first possible solution is for BS $k_1$ to offload all the traffic into BS $k_2$ (i.e., $x_{k_1,n} = D_{k_1,n}$ and $x_{k_2,n} = 0$), such that BS $k_1$ can be turned to sleep mode, for which the total energy consumption of the two BSs is $a_{k_2} (D_{k_1,n} + D_{k_2,n}) + b_{k_2} + c_{k_1}$. Similarly, the other possible solution is $x_{k_1,n} = 0$ and $x_{k_2,n} = D_{k_2,n}$, for which the total energy consumption of the two BSs is $a_{k_1} (D_{k_1,n} + D_{k_2,n}) + b_{k_1} + c_{k_2}$. By comparing the two values of energy consumption, the solution in Case I is proved.

Next, consider Case II when $D^{\max}_{k_i} \le D_{k_1,n} + D_{k_2,n} \le D^{\max}_{k_{\bar \imath}}$. In this case, it can be shown that it suffices for us to consider the following two possible solutions. The first possible solution is for BS $k_i$ to offload all its traffic to BS $k_{\bar\imath}$ (i.e., $x_{k_i,n}= D_{k_i,n}$ and $x_{k_{\bar\imath},n}=0$), such that BS $k_i$ can be turned into sleep, for which the total energy consumption of the two BSs is $a_{k_{\bar\imath}} (D_{k_1,n} + D_{k_2,n}) + b_{k_{\bar\imath}} + c_{k_i}$. The other possible solution is for BS $k_{\bar\imath}$ to offload an $x_{k_{\bar\imath},n} = D^{\max}_{k_i} - D_{k_{\bar\imath},n}$ amount of traffic to BS $k_i$ (and $x_{k_i,n} =0$), such that BS $k_i$ is fully loaded, for which the total energy consumption of the two BSs is $a_{k_i} D^{\max}_{k_i} + b_{k_i} + a_{k_{\bar\imath}} (D_{k_1,n} + D_{k_2,n} - D^{\max}_{k_i}) + b_{k_{\bar\imath}}$. By comparing the two possible solutions, we have the optimal solution in Case II.

Finally, consider Case III with $D_{k_1,n} + D_{k_2,n} > D^{\max}_{k_i},\forall i\in\{1,2\}$. In this case, no BSs can be turned into sleep, and thus it is desirable for the two BSs to offload their wireless traffic to the BS with lower transmission power consumption (accordingly, BS $k_i$ with $a_{k_i} \le a_{k_{\bar \imath}}$). Based on this observation, the solution in Case III can be easily obtained.

By combining the above three cases, this proposition is verified.

\subsection{Proof of Proposition \ref{proposition:5.1}}\label{app:D}

First, we show the optimal solution to problem (\ref{eqn:NBS}). It is evident that under $G_n = G_{1,n} + G_{2,n}$, the optimal solution to (\ref{eqn:problem2}) (i.e., $G^{B**}_{n}(G_{n})$, $G^{S**}_{n}(G_{n})$, and $\{x_{k_i,n}^{**}\}$) is also feasible solutions of $G_{n}^B$, $G_{n}^S$, and $\{x_{k_i,n}\}$ to problem (\ref{eqn:NBS}). Furthermore, from (\ref{eqn:cooperative2}) and (\ref{eqn:cooperative3}), we have $\sum_{i=1}^2  \bar C_{i,n}(G_{i,n},G_{i,n}^B, G_{i,n}^S, \{\pi_{i,n}\}) =  C_{{\rm TC},n}(G_{n},G_{n}^B, G_{n}^S)$. By using this together with the fact that the total energy cost can be perfectly shared between the two MNOs due to the inter-MNO payment $\{\pi_{i,n}\}$, it follows that $G^{B**}_{n}(G_{n})$, $G^{S**}_{n}(G_{n})$, and $\{x_{k_i,n}^{**}\}$, which minimize the total energy cost $\sum_{i=1}^2  \bar C_{i,n}(G_{i,n},G_{i,n}^B, G_{i,n}^S, \{\pi_{i,n}\})$ of the two MNOs, are indeed optimal for problem (\ref{eqn:NBS}), i.e., $G_{n}^{B\star} = G^{B**}_{n}(G_n)$, $G_{n}^{S\star} = G^{S**}_{n}(G_{n})$, and $x^\star_{k_i,n} = x^{**}_{k_i,n}, \forall k\in\{1,\ldots,K\},i\in\{1,2\}$. Substituting $G^{B\star}_{n}$, $G^{S\star}_{n}$, and $\{x_{k_i,n}^{\star}\}$ into problem (\ref{eqn:NBS}), it can be easily shown that the optimization of $\{G_{i,n}^{B}\}, \{G_{i,n}^{S}\}$, and $\{\pi_{i,n}\}$ corresponds to a convex optimization problem.

Next, we obtain the resulting energy cost $\{\bar C^{\star}_{i,n}(\{G_{i,n}\})\}$. As shown above, at the Nash bargaining solution, the resulting total energy cost at the two MNOs is equal to that achieved in the full cooperation (i.e., $C^{\star}_{1,n}(\{G_{i,n}\}) + C^{\star}_{2,n}(\{G_{i,n}\}) = C^{**}_{{\rm TC},n}(G_{n})$ with $G_n = G_{1,n} + G_{2,n}$), and the total energy cost $C^{**}_{{\rm TC},n}(G_{n})$ can be perfectly shared between the two MNOs. Therefore, it is evident from problem (\ref{eqn:NBS}) that $\left\{\bar C^{\star}_{i,n}(\{G_{i,n}\})\right\}_{i=1}^2$ can be obtained as the optimal solution of $\{\bar C_{i,n}\}_{i=1}^2$ to the following convex optimization problem:
\begin{align*}
\max_{\{\bar C_{i,n}\}} ~&(C_{1,n}^{*} (G_{1,n}) - \bar C_{1,n}) \cdot (C_{2,n}^{*} (G_{2,n}) - \bar C_{2,n}) \\
\mathrm{s.t.}~&\bar C_{1,n} +\bar C_{2,n} = C_{{\rm TC},n}^{**} (G_{n})\\
&C_{i,n}^{*} (G_{i,n}) -\bar C_{i,n} \ge 0,\forall i\in\{1,2\}.
\end{align*}
Since $C_{{\rm TC},n}^{**} (G_{n}) \le C_{1,n}^{*} (G_{1,n}) + C_{2,n}^{*} (G_{2,n})$ always holds thanks to the cost reduction achieved by load sharing, (\ref{eqn:36}) can be obtained by solving the above problem. Therefore, this proposition is proved.

\end{document}